\documentclass[journal]{IEEEtran}

\ifCLASSINFOpdf
\usepackage[pdftex]{graphicx}
\DeclareGraphicsExtensions{.pdf,.jpeg,.png,.ps}
\else
\fi

\usepackage{amsmath,amsthm, amssymb, bbding}
\usepackage{mathtools}
\usepackage{pifont}
\usepackage{dsfont}
\usepackage{nicefrac}

\usepackage[hidelinks]{hyperref}

\usepackage{algorithmic}
\usepackage{flushend}

\usepackage{array}

\usepackage{footmisc}

\usepackage{paralist}

\usepackage{tabularx,booktabs}
\usepackage{threeparttable}
\usepackage{multirow}

\ifCLASSOPTIONcompsoc
\usepackage[nocompress]{cite}
\else
\usepackage{cite}
\fi

\ifCLASSOPTIONcompsoc
\usepackage[caption=false,font=footnotesize,labelfont=sf,textfont=sf]{subfig}
\else
\usepackage[caption=false,font=footnotesize]{subfig}
\fi

\usepackage[inline]{enumitem}

\theoremstyle{definition}

\newtheorem{theorem}{Theorem}
\newtheorem{lemma}{Lemma}

\newcolumntype{C}{>{\centering\arraybackslash}X}

\begin{document}

\title{Optimizing Linear Correctors: A Tight Output Min-Entropy Bound and Selection Technique}

\author{Milo\v{s}~Gruji\'{c}
        and~Ingrid~Verbauwhede\thanks{The authors are with COSIC, KU Leuven, Kasteelpark Arenberg 10, 3001 Leuven -- Heverlee, Belgium (e-mail: milos.grujic@esat.kuleuven.be; ingrid.verbauwhede@esat.kuleuven.be).}\thanks{This work was partially supported by CyberSecurity Research Flanders with reference number VR20192203, the European Commission through Horizon 2020 research and innovation program under Belfort ERC Advanced Grant 101020005 and Twinning Grant SAFEST 952252, through Horizon Europe program under grant agreement No.101114043 (QSNP) and through Digital Europe Program together with the Belgian Federal Science Policy Office (Belspo) through the Federal restart and transition plan under grant agreement No.101091625 (BE-QCI).}
}

\maketitle

\begin{abstract}
Post-processing of the raw bits produced by a true random number generator (TRNG) is always necessary when the entropy per bit is insufficient for security applications.
In this paper, we derive a tight bound on the output min-entropy of the algorithmic post-processing module based on linear codes, known as linear correctors.
Our bound is based on the codes' weight distributions, and we prove that it holds even for the real-world noise sources that produce independent but not identically distributed bits.
Additionally, we present a method for identifying the optimal linear corrector for a given input min-entropy rate that maximizes the throughput of the post-processed bits while simultaneously achieving the needed security level. 
Our findings show that for an output min-entropy rate of $\mathbf{0.999}$, the extraction efficiency of the linear correctors with the new bound can be up to $\mathbf{130.56\, \%}$ higher when compared to the old bound, with an average improvement of $\mathbf{41.2\, \%}$ over the entire input min-entropy range.
On the other hand, the required min-entropy of the raw bits for the individual correctors can be reduced by up to $\mathbf{61.62\, \%}$.

\end{abstract}

\begin{IEEEkeywords}
Entropy, true random number generator, post-processing, linear correctors.
\end{IEEEkeywords}

\section{Introduction}
\label{Sec:Introduction}

\IEEEPARstart{R}{andom} numbers produced directly by a noise source of a true random number generator (TRNG) -- raw random numbers, are rarely ideal. 
In order to be considered ideal and possess full entropy, random numbers should be independent, identically and uniformly distributed.
However, raw random numbers often display dependencies, biases, and a lack of identical distribution.
Therefore, before using them for critical security and cryptographic applications, these numbers should be subjugated to entropy extraction (post-processing) to increase the entropy content per random bit to an acceptable level.
An important figure-of-merit of the post-processing algorithms is the extraction efficiency, which represents the ratio of the output to the input entropy.
According to the US standard for TRNGs, referred to as entropy sources in the standard, NIST SP 800-90B \cite{Turan2018}, the raw random numbers can be post-processed (conditioned) by either using one of the six vetted conditioning algorithms or by using custom algorithms with appropriate entropy estimation.
On the other hand, German AIS-31 \cite{AIS31, AIS31New}, which has emerged as the leading TRNG standard and evaluation methodology within the European Union \cite{Balasch18}, categorizes post-processing methods into two main types: cryptographic and algorithmic post-processing.

While the main role of cryptographic post-processing is to ensure computational security \cite{AIS31, AIS31New}, it is also used to increase the entropy rate (entropy per bit) of the random numbers.
To achieve this enhancement, it is crucial for the cryptographic post-processing to be compressive. 
The well-understood and widely used cryptographic hash functions and block ciphers, as building blocks of one-way compression functions, can be used for this purpose.
The security and entropy of the output from the cryptographic post-processing can be derived by modeling it as a random mapping, as discussed in \cite{AIS31, AIS31New}.
Since the random mapping behavior is a theoretical idealization, the entropy estimation of the output relies on the computational security of the used underlying cryptographic primitive.
Cryptographic post-processing is not tailored to any specific distribution family of the raw random numbers. 
It can often be attractive from a practical perspective in security systems that already have software or dedicated hardware implementations of cryptographic primitives. 
However, using cryptographic primitives for the sole purpose of post-processing can also be prohibitively expensive.
Most noise sources produce raw numbers at rates significantly lower than the operating frequencies of modern CPUs \cite{Yang2018, Petura2016, Ma19, Johnston18}. Consequently, the cryptographic post-processing tasks would require a considerable amount of processing time due to the resulting latency.
In digital platforms with dedicated cryptographic accelerators, all non-TRNG applications that require their use would be precluded from employing them during the post-processing.
Further, performing cryptographic operations can be power- or energy-expensive, thereby increasing the overall cost of randomness.

Algorithmic post-processing entails using straightforward and lightweight functions often adapted to the stochastic model of the noise source and the family of raw bit distributions \cite{AIS31, AIS31New}.
Unlike cryptographic post-processing, the algorithmic methods provide information-theoretical security and the output entropy can often be precisely determined.
This post-processing is inherently future-proof when used appropriately, as new and improved cryptanalytic techniques cannot compromise its security. 
For the noise sources that produce independent and identically distributed (IID) bits, the well-known Von Neumann unbiasing \cite{VonNeumann1951} can be used as algorithmic post-processing to obtain the full entropy output.
While Von Neumann's procedure's maximum extraction efficiency of only 0.25 can be increased by its generalizations -- Peres' \cite{Peres1992} and Elias' \cite{Elias1972} unbiasing methods, this comes at a much greater computational cost.
Additional practical disadvantages of these constructions are their variable output rate and the strict IID requirement, which might be impossible to achieve with real-world TRNGs.
Another commonly used algorithmic post-processing method is the simple XOR function of $n$ consecutive bits, which reduces the bias of independent but not necessarily identically distributed raw bits at the cost of $n$-fold throughput reduction \cite{Davies2002}.
While this post-processing can never achieve full entropy of the output bits, it can increase the entropy rate to the desired amount, has a fixed output rate and very low implementation costs.

In \cite{Dichtl2007}, Dichtl proposed several XOR-based post-processing constructions for IID bits with higher extraction efficiency than the basic XOR function due to the reuse of input bits.
These constructions were later formalized as \textit{linear correctors} by Lacharme in \cite{Lacharme2008,Lacharme2009}, who also gave a lower bound on the min-entropy of their output.
Linear correctors are represented by the mappings of the form:
\begin{equation}
	\boldsymbol{Y}^{k \times 1} = \boldsymbol{G}^{k \times n} \boldsymbol{X}^{n \times 1},
\end{equation}
where $\boldsymbol{X}^{n \times 1}$ and $\boldsymbol{Y}^{k \times 1}$ are column vectors of $n$ input and $k$ output bits, respectively, $\boldsymbol{G}^{k \times n}$ is a generator matrix of a binary linear code with minimum distance $d$ and multiplication is performed in the Galois field of size 2. 
If all input bits have bias $\delta$, then the lower bound on the min-entropy of the output of the linear corrector can be derived as \cite{Lacharme2008}:
\begin{equation}\label{eq:Lacharme_ineq}
	\mathrm{H}^{out,\, tot}_{\infty} \geq k - \log_{2}\left(1+\delta^{d}2^{k+d}\right).
\end{equation}
In subsequent works \cite{Tomasi2017, Grujic2022TROT}, it was shown that the linear correctors could also be used on the independent raw bits that are not identically distributed.
A slightly modified version of Lacharme's bound, which includes a lower bound on min-entropy of independent raw bits $\mathrm{H}^{in}_{\infty}$, was given in \cite{Grujic2022TROT}:
\begin{equation}\label{eq:Lacharme_ineq_mod}
	\mathrm{H}^{out,\, tot}_{\infty} \geq k -  \log_{2} \left(1 + \left(2^{1-\mathrm{H}^{in}_{\infty}}-1\right)^d \cdot 2^k\right).
\end{equation}
Linear correctors are recognized by the RISC-V consortium \cite{Zeh2021, Saarinen2022} as a form of admissible non-cryptographic post-processing and are recommended to be used in several recent TRNG designs \cite{Ugajin17, Ali19, Park20, Lyp21, Massari2022, Grujic2022TROT}. 
They represent an attractive post-processing method due to a significantly smaller hardware footprint compared to cryptographic post-processing \cite{Kwok2011}, the ability to deal with not identically distributed raw random bits and higher extraction efficiency than simple XOR function \cite{Dichtl2007, Lacharme2008}.
Refining the output min-entropy bound of the corrector can prevent the unnecessary dissipation of entropy from raw bits during the post-processing stage, thereby enhancing the performance of TRNG designs that incorporate linear correctors.

\subsection{Our Contributions}\label{Subsec:Contributions}

In this work, we noticeably improve Lacharme's previously established min-entropy bound of the linear corrector's output.
The improvement is achieved by first establishing new relations between the probabilities of a linear code and its cosets.
These relations are then used to gain new insights into the connection between the weight distribution of a binary linear code and the linear corrector's output probabilities.
We show that our new bound is also suitable for TRNGs whose noise sources produce independent and non-identically distributed raw bits.
To demonstrate the applicability of this newly established result, we devise an optimization procedure to select linear correctors that achieve the best trade-off between the necessary input min-entropy rate and the throughput reduction to obtain the desired output min-entropy rate.
We leverage the existing knowledge of the best known linear codes and known weight distributions to find the optimal performing linear correctors.
Our newly introduced bound enables us to find linear correctors that are up to $130.56 \%$ more efficient in entropy extraction compared to those derived from the previous bound for an equivalent input min-entropy. 
Across the entire examined input min-entropy range, the new bound averages an enhancement in extraction efficiency by $41.2 \%$.
We have made the list of optimal performing correctors according to the new bound available at \cite{OurGithub}, along with the weight distributions of their corresponding codes and the input min-entropies required to use them. 
This resource is intended to help TRNG designers in selecting appropriate post-processing techniques and to facilitate the reproduction of our work.

 \section{Preliminaries}
\label{Sec:Preliminaries}

In this section, we introduce notation, basic definitions and necessary background in coding theory. For a more in-depth treatment of the coding theory fundamentals, we recommend referring to \cite{MacWilliams1977} and \cite{Lin2004} along with their respective references.

\subsection{Notations and Definitions}\label{Subsec:Preliminaries_notations}

We denote binary vectors with bold lowercase italic letters and matrices with bold uppercase italic letters.
Calligraphic uppercase letters represent random variables, while
the uppercase italic letters are reserved for denoting sets.
The $i$-th bit from the left of an $n$-bit vector $\boldsymbol{x}$ is denoted as $\boldsymbol{x}\left[i\right]$ and is referred to as the $i$ coordinate of $\boldsymbol{x}$.
The Hamming weight of a binary vector $\boldsymbol{x}$ is the number of coordinates of $\boldsymbol{x}$ equal to 1 and we denote it by HW$\left(\boldsymbol{x}\right)$. We use $\boldsymbol{1_{l_0}}$ to denote a bit vector characterized by having a value of 1 exclusively at the $l_0$ coordinate and zeros elsewhere.
The probability of an event is denoted with $\mathds{P}\left[\cdot\right]$.
Let $S$ be some set of $n$-bit vectors $\boldsymbol{x}$, which are realizations of an $n$-bit discrete random variable $\mathcal{X}$ with independent coordinates.
The probability of set $S$ is then defined as the sum of the occurrence probabilities of its element vectors, i.e., 
\begin{equation}
	\mathds{P}\left[S\right] = \sum_{\boldsymbol{x} \in S} \prod_{i=0}^{n-1} \left(\boldsymbol{x}\left[i\right] p_{i} + \left(1 - \boldsymbol{x}\left[i\right]\right) \left(1 - p_{i}\right) \right),
\end{equation}
where $p_i = \mathds{P}\left[\mathcal{X}\left[i\right] = 1\right]$, $0 \leq i \leq n-1$, and $p_i$ is called the 1-probability of bit in coordinate $i$.
$\mathcal{X}$ is an independent and identically distributed (IID) random variable (source) only when $p_i$ is identical for all $n$ bits of $\mathcal{X}$.
In this work, we use min-entropy as a post-processing performance measure, as it is the most conservative uncertainty quantity and is used by both NIST SP 800-90B \cite{Turan2018} and the latest version of AIS-31 standards\cite{AIS31New}.
The min-entropy of a discrete random variable $\mathcal{R}$, with the outcomes from the set $R$, is defined as 
\begin{equation}
	\mathrm{H}_{\infty} = - \log_{2}\left(\underset{r \in R}{\text{max}}\, \mathds{P}\left[\mathcal{R} = r\right]\right).
\end{equation}
In this work, we formally define the extraction efficiency of the post-processing algorithm as
\begin{equation}\label{eq:extraction_efficiency}
    \eta = \frac{\mathrm{H}^{out,\, tot}_{\infty}}{n\,\mathrm{H}^{in}_{\infty}},
\end{equation}
where $\mathrm{H}^{out,\, tot}_{\infty}$ is the total entropy at the output, $n$ is the number of input raw bits and $\mathrm{H}^{in}_{\infty}$ is the lower bound on the min-entropy rate of the raw bits.
We also define post-processing throughput reduction as the ratio of the number of input bits versus the number of output bits.

\subsection{Coding Theory}\label{Subsec:Preliminaries_coding_theory}

A binary linear code $C_0$ of length $n$ and dimension $k$ is a $k$-dimensional subspace of the vector space $\mathds{F}_2^{n}$. 
Hence, $C_0$ is a set of order $2^k$ of $n$-bit row vectors called codewords that form a group under the operation of bitwise modulo 2 addition ($\oplus$).
A minimum distance of a binary linear code is the smallest Hamming weight of the non-zero codewords.
A binary linear code $C_0$ of length $n$, dimension $k$ and minimum distance $d$ is called $\left[n,k,d\right]$-code or just $\left[n,k\right]$-code when properties of a code can be generalized independently of $d$.
Quantity $\nicefrac{k}{n}$ is called the code rate.

\noindent \textit{Example:} Consider a $\left[3,2\right]$-code. 
Here, $n=3$ and $k=2$. 
A potential code could be $C^{A}_0 = \left\{{000}, {110}, {101}, {011}\right\}$, which forms a 2-dimensional subspace in $\mathds{F}_2^{3}$. 
The minimum distance of this code is 2, as that is the smallest Hamming weight among the non-zero codewords ${110}$, ${101}$, and ${011}$.
Another potential $\left[3,2\right]$-code could be $C^{B}_0 = \left\{{000}, {110}, {100}, {010}\right\}$.
The minimum distance of this code is 1, as that is the smallest Hamming weight among the non-zero codewords ${110}$, ${100}$, and ${010}$.

The list of non-negative integers $\left(\mathfrak{A}_i\right)_{i=0}^n$, where $\mathfrak{A}_i$ is the number of codewords of Hamming weight $i$ in a $\left[n,k\right]$-code $C_0$, is called the weight distribution of the code. 

\noindent \textit{Example:} For the $\left[3,2\right]$-code $C^{A}_0$ provided earlier, the weight distribution is $\mathfrak{A}_0 = 1$, $\mathfrak{A}_1 = 0$, $\mathfrak{A}_2 = 3$ and $\mathfrak{A}_3 = 0$ since there is one codeword of weight 0, zero codewords of weight 1, and three codewords of weight 2.

For any binary linear code and for any given coordinate, either all codewords have a 0 at that coordinate or exactly half of them \cite{Lin2004}.
A generator matrix $\boldsymbol{G}$ of an $\left[n,k\right]$-code $C_0$ is a binary $k \times n$ full rank matrix whose rows are $k$ linearly independent codewords of $C_0$.

\noindent \textit{Example:} Let us consider our $\left[3,2\right]$-code $C^{A}_0$ again. 
When we look at the first coordinate, two codewords have a 1 (${110}, {101}$) and the other two have a 0 (${000}, {011}$).
A possible generator matrix $\boldsymbol{G}$ for this code could be:
$\left(\begin{smallmatrix} 1 & 0 & 1 \\0 & 1 & 1 \\\end{smallmatrix}\right)$.
This matrix represents two linearly independent codewords from $C^{A}_0$.
If we consider $\left[3,2\right]$-code $C^{B}_0$ and look at the third coordinate, we see that all codewords have a 0 at this coordinate.

A full rank $\left(n-k\right) \times n$ binary matrix $\boldsymbol{H}$ such that for all codewords $\boldsymbol{c}$ of an $\left[n,k\right]$-code $C_0$ it holds $\boldsymbol{H}\boldsymbol{c^{\intercal}} = \boldsymbol{0}$ is called a parity-check matrix of $C_0$.
For any $n$-bit vector $\boldsymbol{x}$, the parity-check matrix determines the syndrome of $\boldsymbol{x}$ as $\boldsymbol{s} = \boldsymbol{H} \boldsymbol{x}^{\intercal}$.
A binary linear $\left[n,n-k\right]$-code $C^{\bot}_0$ whose generator matrix is the parity-check matrix of $C_0$ is called the dual code of $C_0$.
$C^{\bot}_0$ is the null space of $C_0$, i.e., for any codeword $\boldsymbol{c}$ of $C_0$ and any codeword $\boldsymbol{c^{\bot}}$ of its dual code $C^{\bot}_0$ it holds 
$\sum_{i=0}^{n-1}\boldsymbol{c}[i]\boldsymbol{c^{\bot}}[i] = 0$, where additions and multiplications are in $\mathds{F}_2$.
The weight distribution of the dual code $\left(\mathfrak{A}^{\bot}_i\right)_{i=0}^n$ is called the dual weight distribution and it is related to the weight distribution $\left(\mathfrak{A}_i\right)_{i=0}^n$ of the $C_0$ code by the MacWilliams identity \cite{Lin2004},\cite{Macwilliams1963theorem}:
\begin{equation}\label{eq:macwilliams_id}
	\sum_{i=0}^{n}\mathfrak{A}_i z^{i} = 2^{k-n} \left(1+z\right)^{n} \sum_{i=0}^{n}\mathfrak{A}^{\bot}_i \left(\frac{1-z}{1+z}\right)^{i}.
\end{equation}

\noindent \textit{Example:} Assuming the generator matrix $\boldsymbol{G}$ mentioned above, a parity-check matrix $\boldsymbol{H}$ for our $\left[3,2\right]$-code $C^{A}_0$ is:
$\left(\begin{smallmatrix}
	1 & 1 & 1 \\
\end{smallmatrix}\right)
$.
This matrix ensures that for all codewords $\boldsymbol{c}$ in $C^{A}_0$, $\boldsymbol{H}\boldsymbol{c^{\intercal}} = \boldsymbol{0}$.
Matrix $\boldsymbol{H}$ is at the same time generator matrix of the dual $\left[3,3-2\right]$ code $C^{A,\bot}_0 = \left\{{000}, {111}\right\}$ with weight distribution $\mathfrak{A}^{\bot}_0 = 1$, $\mathfrak{A}^{\bot}_1 = 0$, $\mathfrak{A}^{\bot}_2 = 0$ and $\mathfrak{A}^{\bot}_3 = 1$.

For a binary linear $\left[n,k\right]$-code $C_0$ and an $n$-bit vector $\boldsymbol{a}$, the set $\left\{\boldsymbol{a} \oplus \boldsymbol{c} \mid \boldsymbol{c} \in C_0 \right\}$ is called a coset of $C_0$.
Two $n$-bit vectors are in the same coset if and only if they have an identical syndrome.
Hence, a syndrome uniquely determines a coset. 
A coset leader is the element with the smallest Hamming weight in its coset. 
If there are multiple elements with the same minimal Hamming weight, any of them can be selected to be the coset leader.
We will also sometimes refer to the set of codewords $C_0$ as a coset, with the all-zero vector being its unique coset leader.
The total number of cosets of an $\left[n, k\right]$-code is $2^{n-k}$, including the set of codewords.

\noindent \textit{Example:} Let us continue with our $\left[3,2\right]$-code $C^{A}_0$ and consider the vector $\boldsymbol{a} = {100}$. 
The coset for this vector will be: $\left\{{100}, {010}, {001}, {111} \right\}$.
This is the result of ${100}$ xored with each codeword in $C^{A}_0$.
The coset leader can be any of the vectors with the smallest Hamming weight. 
In this case, any of the weight 1 vectors ${100}$, ${010}$, or ${001} $ could be chosen.
The total number of cosets of our $\left[3,2\right]$-code $C^{A}_0$ would be $2^{3-2} = 2$, meaning that no other 3-bit vector $\boldsymbol{a}$ produces a new coset.

Since there is an equivalence between binary linear codes and linear correctors \cite{Lacharme2008}, we will sometimes interchangeably use the terms corrector and code.

\section{Previous Work}
\label{Sec:Previous_work_linear_correctors}

The relationship between a code's weight distribution and the output of a linear corrector was first noted by Lacharme in \cite{Lacharme2009}, although the previously established min-entropy bound in \cite{Lacharme2008} was not improved.
Zhou et al. \cite{Zhou2011, Zhou2012} studied the exact, average, and asymptotic performance of linear correctors and more general random binary matrices, but only in terms of their statistical distance from the uniform distribution, without considering the entropy rate. 
In \cite{Kwok2011}, Kwok et al. compared the performance of Von Neumann unbiasing, XOR function, and linear correctors with respect to throughput reduction, post-processed bit bias, and adversarial bias reduction. 
However, their study did not consider the performance of these post-processing techniques for non-identically distributed input bits, nor did it account for the correlation between the output bits of a linear corrector and, therefore, the total entropy of the output. 
In contrast, Meneghetti et al. \cite{Meneghetti2014} and Tomasi et al. \cite{Tomasi2017} provided a bound on the statistical distance of linear correctors' output from the uniform distribution based on the code's weight distribution, and they also determined a lower bound on the Shannon entropy using Sason's theorem \cite{Sason2013}, which relates statistical distance and entropy. 
However, this bound is loose because it relies on the statistical distance bound and does not apply to the min-entropy, which is always lower than the Shannon entropy.

In the following section, we will use and expand on two older results from coding theory to improve Lacharme's bound: Sullivan's subgroup-coset inequality \cite{Sullivan1967} and its generalization by \v{Z}ivkovi\'{c} \cite{Zivkovic1991}. Sullivan showed in \cite{Sullivan1967} that when all coordinate 1-probabilities of $n$-bit vectors are smaller than 0.5, the probability of the set of codewords is the highest among all coset probabilities. 
\v{Z}ivkovi\'{c} later demonstrated in \cite{Zivkovic1991} that this relation also holds for any $q$-ary linear code, where $q$ is a prime power, even when individual coordinate 1-probabilities are different but all smaller than 0.5.

\section{Improving the Min-entropy Bound}
\label{Sec:Improved_bound}

We improve the min-entropy bound for linear correctors by first generalizing Sullivan's subgroup-coset inequality \cite{Sullivan1967} for binary linear codes and cases when the coordinate 1-probabilities are different and not upper limited to 0.5.
First, we recall a lemma from \cite{Sullivan1967} that will also be used in our proofs.
\begin{lemma}[adapted from \cite{Sullivan1967}]
	\label{lemma:subgroup_order}
Let $C_0$ be a binary linear $\left[n, k\right]$-code, and let $\boldsymbol{e}$, HW$\left(\boldsymbol{e}\right) = l$, be a coset leader in some coset of $C_0$.
	Then the code $C^{\prime}_0$, obtained by deleting $l$ coordinates in which $\boldsymbol{e}$ is 1, is a binary linear $\left[n-l, k\right]$-code.
\end{lemma}

We now introduce our first inequality theorem, named the \textit{coset-coset inequality}. 
This theorem establishes a relationship between the probabilities of two distinct cosets belonging to a specific binary linear code. 
It offers a distinctive perspective when compared to the subgroup-coset inequalities proposed by Sullivan and \v{Z}ivkovi\'{c}. 
The proof of this theorem builds upon the foundations laid out in \cite{Sullivan1967} and \cite{Zivkovic1991}.

\begin{theorem}
	\label{thm:generalized_Sullivan}
Let $C_i$, $0 \leq i \leq 2^{n-k}-1$, denote sets of $n$-bit element vectors $\boldsymbol{x}$, which are realizations of the $n$-bit row vector random variable $\mathcal{X}$ with independent coordinates. Let $C_0$ be a binary linear $\left[n, k\right]$-code, and all other $C_i$, $i \not = 0$, are cosets of $C_0$.
	Let $C_{i_{max}}$ denote the set that contains the most probable element vector $\boldsymbol{x_{max}}$ with coordinates
\begin{equation*}
		\boldsymbol{x_{max}}\left[j\right] =
		\begin{cases}
			1,\, \text{if} \,\, 0.5 \leq p_j \leq 1,\\
			0,\, \text{if} \,\, 0 \leq p_j < 0.5,
		\end{cases}
	\end{equation*}
	where $p_j$ is the 1-probability of bit in $j$ coordinate, $0 \leq j \leq n-1$.
Then it holds $\mathds{P}\left[C_{i_{max}}\right] \geq \mathds{P}\left[C_i\right]$, and we call $C_{i_{max}}$ the most probable set.
\end{theorem}

\begin{proof}
	
First, we arrange all possible $2^n$ vectors in the standard array such that the $i$-th row contains elements of the set $C_i$.
	The first entry in each row $\boldsymbol{c_{i,\, 0}}$ is a coset leader $\boldsymbol{e_i}$, i.e., a vector with the lowest weight in the corresponding set, while all other row entries $\boldsymbol{c_{i,\, j}}$ are obtained by adding $\boldsymbol{e_i}$ and the corresponding entry in the 0-th row:  $\boldsymbol{c_{i,\, j}} = \boldsymbol{e_i} \oplus \boldsymbol{c_{0,\, j}}$, $1 \leq j \leq 2^{k}-1$.
Consider now the set that contains the most probable vector $C_{i_{max}}$ with coset leader $\boldsymbol{e_{i_{max}}}$ and some arbitrary but fixed set $C_i$, $C_i \neq C_{i_{max}}$, with coset leader $\boldsymbol{e_{i}}$, as well as their corresponding rows in the standard array.
If $\boldsymbol{e_{i}} \oplus \boldsymbol{e_{i_{max}}}$ is in some set $C_l$, but is not equal to its coset leader $\boldsymbol{e_{l}}$, we rearrange the entries in the $i_{max}$-th row so that for the first entry $\boldsymbol{c_{i_{max},\, 0}}$ we select an element of $C_{i_{max}}$ that is equal to $\boldsymbol{e_{i}} \oplus \boldsymbol{e_{l}}$.
	All other row entries are rearranged so that the $j$-th element is equal to $\boldsymbol{c_{i_{max},\, j}} = \boldsymbol{c_{i_{max},\, 0}} \oplus \boldsymbol{c_{0,\, j}} = \boldsymbol{e_{i}} \oplus \boldsymbol{e_{l}} \oplus \boldsymbol{c_{0,\, j}}$.
	On the other hand, no rearrangements are made if $\boldsymbol{e_{i}} \oplus \boldsymbol{e_{i_{max}}} = \boldsymbol{e_{l}}$ already holds.
After possible rearrangement, any entry in the $i$-th row $\boldsymbol{c_{i,\, j}}$ is related to the entry $\boldsymbol{c_{i_{max},\, j}}$ in the $i_{max}$-th row by relation $\boldsymbol{c_{i,\, j}} = \boldsymbol{c_{i_{max},\, j}} \oplus \boldsymbol{e_l}$.
	Since all entries in the $i$-th and $i_{max}$-th row are also elements of the sets $C_i$ and $C_{i_{max}}$, respectively, this shows that every element in $C_{i_{max}}$ has exactly one corresponding element in $C_i$ from which it differs only in coordinates in which $\boldsymbol{e_{l}}$ is 1.
	We will prove the theorem by double induction over the code dimension $k$, $0 \leq k < n$, and the Hamming weight of the coset leaders HW$\left(\boldsymbol{e_{l}}\right) \geq 1$.

	\noindent \textit{Base case.} For $k=0$ and HW$\left(\boldsymbol{e_{l}}\right) = 1$, we have $C_0 = \left\{\boldsymbol{0}\right\}$, where $\boldsymbol{0}$ is the all-zero vector. Since, in this case, each set contains only one $n$-bit vector, it is clear that the set that includes the most probable vector $\boldsymbol{x_{max}}$ will have probability $\mathds{P}\left[C_{i_{max}}\right] = \mathds{P}\left[\boldsymbol{x_{max}}\right]$ and that $\mathds{P}\left[C_{i_{max}}\right] \geq \mathds{P}\left[C_i\right]$ always holds.

	\noindent \textit{Outer induction hypothesis.} Assume that the theorem is true for all binary linear codes of dimension $k \leq k^{\prime}$ and HW$\left(\boldsymbol{e_{l}}\right) = 1$.

	\noindent \textit{Outer induction step.} We will show that the outer induction hypothesis implies that the theorem also holds for all binary linear codes of dimension $ k = k^{\prime}+1$ and HW$\left(\boldsymbol{e_{l}}\right) = 1$.
	Suppose that $\boldsymbol{e_{l}}$ has 1 in coordinate $l_0$ and let $I_{l_0} = \left\{0,\dots, n-1\right\} \backslash\left\{l_0\right\}$.
We now partition sets $C_i$ and $C_{i_{max}}$ into two subsets, depending on the value in coordinate $l_0$ of their element vectors: $C_{i_{max}}^{l_0,\, b} =\left\{\boldsymbol{x} \in C_{i_{max}} \mid \boldsymbol{x}[l_0] = b \right\} $ and $C_{i}^{l_0,\, b} =\left\{\boldsymbol{x} \in C_{i} \mid \boldsymbol{x}[l_0] = b \right\}$, $b \in \{0,1\}$.
	
	\textit{Case 1a:} Suppose first that $\boldsymbol{x}[l_0] = \hat{b}$ holds for all $\boldsymbol{x} \in C_{i_{max}}$, where $\hat{b}$ is fixed to either 0 or 1.
	Then the order of $C_{i_{max}}^{l_0,\hat{b}}$ is $2^{k^{\prime}+1}$ since $C_{i_{max}}^{l_0,\hat{b}} = C_{i_{max}}$ and $C_{i_{max}}^{l_0, 1-\hat{b}} = \varnothing$.
	Given that the elements in $C_{i}^{l_0, 1-\hat{b}}$ differ from the elements in $C_{i_{max}}^{l_0,\hat{b}}$ only in the $l_0$ coordinate, we have that the order of $C_{i}^{l_0, 1-\hat{b}}$ is also $2^{k^{\prime}+1}$ and $C_{i}^{l_0, 1-\hat{b}} = C_{i}$, while $C_{i}^{l_0,\hat{b}} = \varnothing$.
	Therefore, we can express the probabilities of the sets $C_{i_{max}}$ and $C_i$ as
	\begin{multline}\label{eq:p_c_imax_0_m_1}
		\mathds{P}\left[C_{i_{max}}\right] = \mathds{P}\left[C_{i_{max}}^{l_0,\hat{b}}\right] =\\ \left(\hat{b} p_{l_0} + \left(1 - \hat{b}\right) \left(1 - p_{l_0}\right) \right)\cdot \\ \sum_{\boldsymbol{x} \in C_{i_{max}}} \prod_{i \in I_{l_0}} \left(\boldsymbol{x}\left[i\right] p_{i} + \left(1 - \boldsymbol{x}\left[i\right]\right) \left(1 - p_{i}\right) \right),  
	\end{multline}
	and
	\begin{multline}\label{eq:p_c_i_0_m_1}
		\mathds{P}\left[C_{i}\right] = \mathds{P}\left[C_{i}^{l_0,1-\hat{b}}\right] =\\ \left(\left(1 -\hat{b}\right) p_{l_0} + \hat{b} \left(1 - p_{l_0}\right) \right)\cdot\\ \sum_{\boldsymbol{x} \in C_{i}} \prod_{i \in I_{l_0}} \left(\boldsymbol{x}\left[i\right] p_{i} + \left(1 - \boldsymbol{x}\left[i\right]\right) \left(1 - p_{i}\right) \right).  
	\end{multline}
	Note that
	\begin{multline}
		\sum_{\boldsymbol{x} \in C_{i}} \prod_{i \in I_{l_0}} \left(\boldsymbol{x}\left[i\right] p_{i} + \left(1 - \boldsymbol{x}\left[i\right]\right) \left(1 - p_{i}\right) \right) = \\
		\sum_{\boldsymbol{x} \in C_{i_{max}}} \prod_{i \in I_{l_0}} \left(\boldsymbol{x}\left[i\right] p_{i} + \left(1 - \boldsymbol{x}\left[i\right]\right) \left(1 - p_{i}\right) \right),
	\end{multline}
	holds since the elements in $C_{i}$ and $C_{i_{max}}$ differ only in the $l_0$ coordinate.
	
	\textit{Subcase 1.1a:} For $\hat{b} = 0$, it holds $1 - p_{l_0} > p_{l_0}$ since $0 \leq p_{l_0} < 0.5$, which follows from the fact that $\boldsymbol{x_{max}} \in C_{i_{max}}$ and all vectors in $C_{i_{max}}$ have 0 in coordinate $l_0$ for $\hat{b} = 0$.
	Therefore, from (\ref{eq:p_c_imax_0_m_1}) and (\ref{eq:p_c_i_0_m_1}), we have the inequality
\begin{multline}
		\mathds{P}\left[C_{i_{max}}\right] =\\  \left(1 - p_{l_0}\right) \sum_{\boldsymbol{x} \in C_{i_{max}}} \prod_{i \in I_{l_0}} \left(\boldsymbol{x}\left[i\right] p_{i} + \left(1 - \boldsymbol{x}\left[i\right]\right) \left(1 - p_{i}\right) \right) > \\ \mathds{P}\left[C_{i}\right] = p_{l_0} \sum_{\boldsymbol{x} \in C_{i}} \prod_{i \in I_{l_0}} \left(\boldsymbol{x}\left[i\right] p_{i} + \left(1 - \boldsymbol{x}\left[i\right]\right) \left(1 - p_{i}\right) \right).
	\end{multline}
	
	\textit{Subcase 1.2a:} For $\hat{b} = 1$, all vectors in $C_{i_{max}}$ have 1 in coordinate $l_0$ and $\boldsymbol{x_{max}} \in C_{i_{max}}$.
	Thus, $p_{l_0} \geq 1-p_{l_0}$, since $0.5 \leq p_{l_0} \leq 1$.
	Hence, $\mathds{P}\left[C_{i_{max}}\right] \geq \mathds{P}\left[C_{i}\right]$ holds in this case as well, which can be seen by substituting $\hat{b} = 1$ in (\ref{eq:p_c_imax_0_m_1}) and (\ref{eq:p_c_i_0_m_1}):
	\begin{multline}
		\mathds{P}\left[C_{i_{max}}\right] =\\  p_{l_0} \sum_{\boldsymbol{x} \in C_{i_{max}}} \prod_{i \in I_{l_0}} \left(\boldsymbol{x}\left[i\right] p_{i} + \left(1 - \boldsymbol{x}\left[i\right]\right) \left(1 - p_{i}\right) \right) \geq \\ \mathds{P}\left[C_{i}\right] = \left(1 - p_{l_0}\right) \sum_{\boldsymbol{x} \in C_{i}} \prod_{i \in I_{l_0}} \left(\boldsymbol{x}\left[i\right] p_{i} + \left(1 - \boldsymbol{x}\left[i\right]\right) \left(1 - p_{i}\right) \right).
	\end{multline}

	\textit{Case 2a:} Suppose the values in coordinate $l_0$ are not identical for all vectors in $C_{i_{max}}$.
	The orders of $C_{i_{max}}^{l_0, 0}$, $C_{i_{max}}^{l_0, 1}$, $C_{i}^{l_0, 0}$ and $C_{i}^{l_0, 1}$ are all equal to $2^{k^{\prime}}$.
	We now delete component in coordinate $l_0$ of every element in both $C_{i}$ and $C_{i_{max}}$ and denote the resulting sets by $C^{\overline{l_0}}_{i}$ and $C^{\overline{l_0}}_{i_{max}}$, and the corresponding partitioning subsets by $C^{\overline{l_0},\, 0}_{i}$, $C^{\overline{l_0}, \, 1}_{i}$, $C^{\overline{l_0}, \, 0}_{i_{max}}$ and $C^{\overline{l_0}, \, 1}_{i_{max}}$.
	Since $C_i$ and $C_{i_{max}}$ are either equivalent to $C_0$ or are its proper cosets, from Lemma~\ref{lemma:subgroup_order}, we have that the orders of $C^{\overline{l_0}}_{i}$ and $C^{\overline{l_0}}_{i_{max}}$ are $2^{k^{\prime}+1}$.
	Consequently, the orders of $C^{\overline{l_0}, \, 0}_{i}$, $C^{\overline{l_0}, \, 1}_{i}$, $C^{\overline{l_0}, \, 0}_{i_{max}}$ and $C^{\overline{l_0}, \, 1}_{i_{max}}$ will be $2^{k^{\prime}}$.
	Since the elements in $C_i$ differ from the elements in $C_{i_{max}}$ only in the coordinate $l_0$, it follows that $C^{\overline{l_0}, \, 0}_{i} = C^{\overline{l_0}, \, 1}_{i_{max}}$ and $C^{\overline{l_0}, \, 1}_{i} = C^{\overline{l_0}, \, 0}_{i_{max}}$.
	The set probabilities $\mathds{P}\left[C_{i_{max}}\right]$ and $\mathds{P}\left[C_{i}\right]$ can be expressed as
	\begin{multline}
		\mathds{P}\left[C_{i_{max}}\right] = \mathds{P}\left[C_{i_{max}}^{l_0, 1}\right] + \mathds{P}\left[C_{i_{max}}^{l_0, 0}\right] =\\ 
		p_{l_0}\mathds{P}\left[C^{\overline{l_0}, \, 1}_{i_{max}}\right] + \left(1 - p_{l_0}\right)\mathds{P}\left[C^{\overline{l_0}, \, 0}_{i_{max}}\right]
	\end{multline}
	and
	\begin{multline}
		\mathds{P}\left[C_{i}\right] = \mathds{P}\left[C_{i}^{l_0, 1}\right] + \mathds{P}\left[C_{i}^{l_0, 0}\right] =\\ 
		p_{l_0}\mathds{P}\left[C^{\overline{l_0}, \, 1}_{i}\right] + \left(1 - p_{l_0}\right)\mathds{P}\left[C^{\overline{l_0}, \, 0}_{i}\right] =\\ p_{l_0}\mathds{P}\left[C^{\overline{l_0}, \, 0}_{i_{max}}\right] + \left(1 - p_{l_0}\right)\mathds{P}\left[C^{\overline{l_0}, \, 1}_{i_{max}}\right].
	\end{multline}
	Thus, we obtain
	\begin{multline}\label{eq: p_c_imax-c_i_m_1}
		\mathds{P}\left[C_{i_{max}}\right] - \mathds{P}\left[C_{i}\right] =\\
		p_{l_0}\mathds{P}\left[C^{\overline{l_0}, \, 1}_{i_{max}}\right] - \left(1 - p_{l_0}\right)\mathds{P}\left[C^{\overline{l_0}, \, 1}_{i_{max}}\right]\\
		+ \left(1 - p_{l_0}\right)\mathds{P}\left[C^{\overline{l_0}, \, 0}_{i_{max}}\right] - p_{l_0}\mathds{P}\left[C^{\overline{l_0}, \, 0}_{i_{max}}\right] =\\
		\left(1 - 2p_{l_0}\right)\left(\mathds{P}\left[C^{\overline{l_0}, \, 0}_{i_{max}}\right] - \mathds{P}\left[C^{\overline{l_0}, \, 1}_{i_{max}}\right]\right).
	\end{multline}
	
	\textit{Subcase 2.1a:} If $\boldsymbol{x_{max}} \in C_{i_{max}}^{l_0, 0}$, then the most probable ${\left(n-1\right)}$-bit vector obtained from $\boldsymbol{x_{max}}$ by deleting its $l_0$ coordinate $\boldsymbol{x^{\overline{l_0}}_{max}} = \left(\boldsymbol{x_{max}}\left[0\right]...\, \boldsymbol{x_{max}}\left[l_0-1\right] {\boldsymbol{x_{max}}}\left[l_0+1\right] ... \, \boldsymbol{x_{max}}\left[n-1\right] \right)$ will be in the subset $C^{\overline{l_0},\, 0}_{i_{max}}$. 
	Hence, from the induction hypothesis $\mathds{P}\left[C^{\overline{l_0},\, 0}_{i_{max}}\right] \geq \mathds{P}\left[C^{\overline{l_0},\, 1}_{i_{max}}\right]$.
	We note that since $\boldsymbol{x_{max}} \in C_{i_{max}}^{l_0, 0}$, we have $0 \leq p_{l_0} < 0.5$, and thus, $1 - 2p_{l_0} > 0$.
	Based on this observation and the outer induction hypothesis, we have that both multiplication terms in the last line of (\ref{eq: p_c_imax-c_i_m_1}) are non-negative, implying that $\mathds{P}\left[C_{i_{max}}\right] \geq \mathds{P}\left[C_{i}\right]$.
	
	\textit{Subcase 2.2a:} If $\boldsymbol{x_{max}} \in C_{i_{max}}^{l_0, 1}$, then $\boldsymbol{x^{\overline{l_0}}_{max}}$ will be an element of the subset $C^{\overline{l_0},\, 1}_{i_{max}}$. 
	From the induction hypothesis, in this case, we have $\mathds{P}\left[C^{\overline{l_0},\, 1}_{i_{max}}\right] \geq \mathds{P}\left[C^{\overline{l_0},\, 0}_{i_{max}}\right]$.
	Furthermore, since $\boldsymbol{x_{max}} \in C_{i_{max}}^{l_0, 1}$, we have $0.5 \leq p_{l_0} \leq 1$, and thus, $1 - 2p_{l_0} \leq 0$.
	Therefore, both terms in the last line of (\ref{eq: p_c_imax-c_i_m_1}) are non-positive, implying that their product is non-negative, and $\mathds{P}\left[C_{i_{max}}\right] \geq \mathds{P}\left[C_{i}\right]$ holds in this case as well.

	By induction, the theorem is true for all binary linear codes' dimensions $k$, $0 \leq k < n$, and HW$\left(\boldsymbol{e_l}\right) = 1$.
	
\noindent \textit{Inner induction hypothesis.} Assume that the theorem holds for all binary linear codes of dimension $k$ and HW$\left(\boldsymbol{e_l}\right)$ values not greater than $m$.
	
	\noindent \textit{Inner induction step.} We proceed with the second induction step by showing that the inner induction hypothesis implies that the theorem holds for HW$\left(\boldsymbol{e_l}\right) = m+1$ and all binary linear codes of dimension $k$.
	Let $l_m$ be one of the $m+1$ possible positions in which $\boldsymbol{e_l}$ has 1, and let $I_{l_m} = \left\{0,\dots, n-1\right\} \backslash\left\{l_m\right\}$.
We separate all elements in both $C_i$ and $C_{i_{max}}$ into two subsets according to their coordinate value in coordinate $l_m$: $C^{l_m,\, b}_{i_{max}} = \left\{\boldsymbol{x} \in C_{i_{max}} \mid \boldsymbol{x}[l_m] = b\right\}$ and $C^{l_m,\, b}_{i} = \left\{\boldsymbol{x} \in C_{i} \mid \boldsymbol{x}[l_m] = b\right\}$, $b \in \left\{0,1\right\}$.
	Let $C^{\overline{l_m}}_{i_{max}}$ and $C^{\overline{l_m}}_{i}$ be sets obtained from $C_{i_{max}}$ and $C_{i}$ by removing the component in coordinate $l_m$ in all vectors in both sets.
	According to Lemma~\ref{lemma:subgroup_order}, the orders of $C^{\overline{l_m}}_{i_{max}}$ and $C^{\overline{l_m}}_{i}$ will remain $2^{k}$ and the elements in $C^{\overline{l_m}}_{i}$ will differ from the elements in $C^{\overline{l_m}}_{i_{max}}$ in coordinates in which vector $\boldsymbol{e^{\overline{l_m}}_l} = (\boldsymbol{e_l}[0] ...\, \boldsymbol{e_l}[l_m-1], \boldsymbol{e_l}[l_m+1] ...\, \boldsymbol{e_l}[n-1])$ is 1.
	Since the most probable ${\left(n-1\right)}$-bit vector $\boldsymbol{x^{\overline{l_m}}_{max}} = (\boldsymbol{x_{max}}[0], ...\, \boldsymbol{x_{max}}[l_m-1], \boldsymbol{x_{max}}[l_m+1], ...\, \boldsymbol{x_{max}}[n-1] )$ will be in set $C^{\overline{l_m}}_{i_{max}}$ and HW$\left(\boldsymbol{e^{\overline{l_m}}_l}\right) = m$, by the inner induction hypothesis, we obtain 
	\begin{multline}\label{eq:c_lm_imax_ineq_c_lm_i}
		\mathds{P}\left[C^{\overline{l_m}}_{i_{max}}\right] = \sum_{\boldsymbol{x} \in C_{i_{max}}} \prod_{i \in I_{l_m}} \left(\boldsymbol{x}\left[i\right] p_{i} + \left(1 - \boldsymbol{x}\left[i\right]\right) \left(1 - p_{i}\right) \right) \\ \geq \mathds{P}\left[C^{\overline{l_m}}_{i}\right] =  \sum_{\boldsymbol{x} \in C_{i}} \prod_{i \in I_{l_m}} \left(\boldsymbol{x}\left[i\right] p_{i} + \left(1 - \boldsymbol{x}\left[i\right]\right) \left(1 - p_{i}\right) \right).
	\end{multline}
	
	\textit{Case 1b:} Suppose $\boldsymbol{x}[l_m] = \hat{b}$ holds for all $\boldsymbol{x} \in C_{i_{max}}$, where $\hat{b}$ is fixed to either a 0 or a 1.
	The order of $C^{l_m,\hat{b}}_{i_{max}}$ is then $2^{k}$ and $C^{l_m,1-\hat{b}}_{i_{max}} = \varnothing$.
	Since $\boldsymbol{e_l}$ has 1 in coordinate $l_m$, all vectors in $C_{i}$ will have $1-\hat{b}$ in coordinate $l_m$.
	Hence, the order of $C^{l_m,1-\hat{b}}_{i}$ is also $2^{k}$ and $C^{l_m,\hat{b}}_{i} = \varnothing$. 
	For the probabilities of sets $C_{i_{max}}$ and $C_i$, we have
	\begin{multline}\label{eq:p_c_imax_lm}
		\mathds{P}\left[C_{i_{max}}\right] = \mathds{P}\left[C^{l_m,\hat{b}}_{i_{max}}\right] =\\
		\left(\hat{b} p_{l_m} + \left(1 - \hat{b}\right) \left(1 - p_{l_m}\right) \right)\cdot\\ \sum_{\boldsymbol{x} \in C_{i_{max}}} \prod_{i \in I_{l_m}} \left(\boldsymbol{x}\left[i\right] p_{i} + \left(1 - \boldsymbol{x}\left[i\right]\right) \left(1 - p_{i}\right) \right),
	\end{multline} 
	and
	\begin{multline}\label{eq:p_c_i_lm}
		\mathds{P}\left[C_{i}\right] = \mathds{P}\left[C^{l_m,1-\hat{b}}_{i}\right] =\\
		\left(\left(1 -\hat{b}\right) p_{l_m} + \hat{b} \left(1 - p_{l_m}\right) \right)\cdot\\ \sum_{\boldsymbol{x} \in C_{i}} \prod_{i \in I_{l_m}} \left(\boldsymbol{x}\left[i\right] p_{i} + \left(1 - \boldsymbol{x}\left[i\right]\right) \left(1 - p_{i}\right) \right).
	\end{multline} 
	By substituting $\mathds{P}\left[C^{\overline{l_m}}_{i_{max}}\right]$ and $\mathds{P}\left[C^{\overline{l_m}}_{i}\right]$ from (\ref{eq:c_lm_imax_ineq_c_lm_i})  in (\ref{eq:p_c_imax_lm}) and (\ref{eq:p_c_i_lm}), and then subtracting $\mathds{P}\left[C_{i}\right]$ from $\mathds{P}\left[C_{i_{max}}\right]$, we obtain
	\begin{multline}\label{eq:p_cimax_minus_p_ci_equal}
		\mathds{P}\left[C_{i_{max}}\right] - \mathds{P}\left[C_{i}\right] =\\ \hat{b} \left(p_{l_m} \mathds{P}\left[C^{\overline{l_m}}_{i_{max}}\right] - \left(1 -p_{l_m}\right) \mathds{P}\left[C^{\overline{l_m}}_{i}\right] \right) +\\ \left(1 -\hat{b}\right)\left(\left(1 -p_{l_m}\right) \mathds{P}\left[C^{\overline{l_m}}_{i_{max}}\right] - p_{l_m} \mathds{P}\left[C^{\overline{l_m}}_{i}\right] \right).
	\end{multline}
	
	\textit{Subcase 1.1b:} For $\hat{b} = 0$, since $\boldsymbol{x_{max}} \in C_{i_{max}}$, we have $0 \leq p_{l_m} < 0.5$, thus, $\left(1 -p_{l_m}\right) > p_{l_m}$.
	Equation (\ref{eq:p_cimax_minus_p_ci_equal}) then becomes
	\begin{multline}\label{eq:p_c_imax_minus_p_c_i_x_0}
		\mathds{P}\left[C_{i_{max}}\right] - \mathds{P}\left[C_{i}\right] =\\ \left(1 -p_{l_m}\right) \mathds{P}\left[C^{\overline{l_m}}_{i_{max}}\right] - p_{l_m} \mathds{P}\left[C^{\overline{l_m}}_{i}\right].
	\end{multline}
	By multiplying both sides of (\ref{eq:c_lm_imax_ineq_c_lm_i}) by $\left(1 -p_{l_m}\right)$ and combining this result with $\left(1 -p_{l_m}\right) > p_{l_m}$, we have the inequality
	\begin{equation}
		\left(1 -p_{l_m}\right)\mathds{P}\left[C^{\overline{l_m}}_{i_{max}}\right] \geq \left(1 -p_{l_m}\right)\mathds{P}\left[C^{\overline{l_m}}_{i}\right] > p_{l_m}\mathds{P}\left[C^{\overline{l_m}}_{i}\right].
	\end{equation} 
	From the preceding inequality and (\ref{eq:p_c_imax_minus_p_c_i_x_0}), it holds $\mathds{P}\left[C_{i_{max}}\right] > \mathds{P}\left[C_{i}\right]$.	
	
	\textit{Subcase 1.2b:} Similarly, for $\hat{b} = 1$, we have $0.5 \leq p_{l_m} \leq 1$, thus, $p_{l_m} \geq \left(1 -p_{l_m}\right)$ and (\ref{eq:p_cimax_minus_p_ci_equal}) becomes
	\begin{multline}\label{eq:p_c_imax_minus_p_c_i_x_1}
		\mathds{P}\left[C_{i_{max}}\right] - \mathds{P}\left[C_{i}\right] =\\ p_{l_m} \mathds{P}\left[C^{\overline{l_m}}_{i_{max}}\right] - \left(1 -p_{l_m}\right) \mathds{P}\left[C^{\overline{l_m}}_{i}\right].
	\end{multline}
	By multiplying both sides of (\ref{eq:c_lm_imax_ineq_c_lm_i}) by $\left(1 -p_{l_m}\right)$ and combining this result with the inequality $p_{l_m} \geq \left(1 -p_{l_m}\right)$, we obtain
	\begin{equation}\label{eq:ineq_x_1_p_cimax_ci}
		p_{l_m}\mathds{P}\left[C^{\overline{l_m}}_{i_{max}}\right] \geq \left(1 -p_{l_m}\right)\mathds{P}\left[C^{\overline{l_m}}_{i_{max}}\right] \geq \left(1 -p_{l_m}\right)\mathds{P}\left[C^{\overline{l_m}}_{i}\right].
	\end{equation} 
	From (\ref{eq:p_c_imax_minus_p_c_i_x_1}) and (\ref{eq:ineq_x_1_p_cimax_ci}), it follows that $\mathds{P}\left[C_{i_{max}}\right] \geq \mathds{P}\left[C_{i}\right]$ holds in this case as well.
	
	\textit{Case 2b:} Suppose that $\boldsymbol{x}\left[l_m\right]$ is not identical for all $\boldsymbol{x} \in C_{i_{max}}$.
Let $C^{\overline{l_m},\, b}_{i_{max}}$ and $C^{\overline{l_m},\, b}_{i}$, $b \in \{0,1\}$, be subsets of $C^{\overline{l_m}}_{i_{max}}$ and $C^{\overline{l_m}}_{i}$, respectively, obtained from $C^{l_m, b}_{i_{max}}$ and $C^{l_m, b}_{i}$ by deleting the $l_0$ coordinate in the element vectors.
We can express the probabilities of sets $C^{\overline{l_m}}_{i_{max}}$ and $C^{\overline{l_m}}_{i}$ as  $\mathds{P}\left[C^{\overline{l_m}}_{i_{max}}\right] = \mathds{P}\left[C^{\overline{l_m},0}_{i_{max}}\right] + \mathds{P}\left[C^{\overline{l_m},1}_{i_{max}}\right]$ and $\mathds{P}\left[C^{\overline{l_m}}_{i}\right] = \mathds{P}\left[C^{\overline{l_m},0}_{i}\right] + \mathds{P}\left[C^{\overline{l_m},1}_{i}\right]$, respectively, and rewrite (\ref{eq:c_lm_imax_ineq_c_lm_i}) as
\begin{multline}\label{ineq:second_ind_hyp_x_not_eq}
		\mathds{P}\left[C^{\overline{l_m}}_{i_{max}}\right] = \mathds{P}\left[C^{\overline{l_m},0}_{i_{max}}\right] + \mathds{P}\left[C^{\overline{l_m},1}_{i_{max}}\right]\\ \geq \mathds{P}\left[C^{\overline{l_m}}_{i}\right] = \mathds{P}\left[C^{\overline{l_m},0}_{i}\right] + \mathds{P}\left[C^{\overline{l_m},1}_{i}\right].
	\end{multline}
	The probabilities $\mathds{P}\left[C_{i_{max}}\right]$ and $\mathds{P}\left[C_{i}\right]$ can be expressed as
	\begin{multline}\label{eq:p_cimax_x_not_equal}
		\mathds{P}\left[C_{i_{max}}\right] = \mathds{P}\left[C^{l_m, 0}_{i_{max}}\right] + \mathds{P}\left[C^{l_m, 1}_{i_{max}}\right] =\\ p_{l_m}\mathds{P}\left[C^{\overline{l_m},1}_{i_{max}}\right] + \left(1 -p_{l_m}\right)\mathds{P}\left[C^{\overline{l_m},0}_{i_{max}}\right] 
	\end{multline}
	and
	\begin{multline}\label{eq:p_ci_x_not_equal}
		\mathds{P}\left[C_{i}\right] = \mathds{P}\left[C^{l_m, 0}_{i}\right] + \mathds{P}\left[C^{l_m, 1}_{i}\right] =\\ p_{l_m}\mathds{P}\left[C^{\overline{l_m},1}_{i}\right] + \left(1 -p_{l_m}\right)\mathds{P}\left[C^{\overline{l_m},0}_{i}\right].
	\end{multline}
	By subtracting (\ref{eq:p_ci_x_not_equal}) from (\ref{eq:p_cimax_x_not_equal}), we obtain
	\begin{multline}\label{eq:p_cimax_min_p_ci}
		\mathds{P}\left[C_{i_{max}}\right] - \mathds{P}\left[C_{i}\right] = \\
		\left(1 -p_{l_m}\right)\left(\mathds{P}\left[C^{\overline{l_m},0}_{i_{max}}\right] - \mathds{P}\left[C^{\overline{l_m},0}_{i}\right]\right)\\ - p_{l_m}\left(\mathds{P}\left[C^{\overline{l_m},1}_{i}\right] - \mathds{P}\left[C^{\overline{l_m},1}_{i_{max}}\right]\right). 
	\end{multline}
	
	\textit{Subcase 2.1b:} First, suppose that $\boldsymbol{x_{max}}[l_m] = 0$, i.e., $\boldsymbol{x_{max}} \in C^{l_m, 0}_{i_{max}}$.
	This implies $0 \leq p_{l_m} < 0.5$ and $\left(1 - p_{l_m}\right) > p_{l_m}$.
	By multiplying both sides of (\ref{ineq:second_ind_hyp_x_not_eq}) by $\left(1 - p_{l_m}\right)$ and rearranging the terms, we have
	\begin{multline}\label{ineq:peniultimum_thm}
		\left(1 - p_{l_m}\right)\left(\mathds{P}\left[C^{\overline{l_m},0}_{i_{max}}\right] - \mathds{P}\left[C^{\overline{l_m},0}_{i}\right]\right) \geq \\ \left(1 - p_{l_m}\right)\left(\mathds{P}\left[C^{\overline{l_m},1}_{i}\right] - \mathds{P}\left[C^{\overline{l_m},1}_{i_{max}}\right]\right) > \\
		p_{l_m}\left(\mathds{P}\left[C^{\overline{l_m},1}_{i}\right] - \mathds{P}\left[C^{\overline{l_m},1}_{i_{max}}\right]\right),
	\end{multline}
	where the last inequality comes from  $\left(1 - p_{l_m}\right) > p_{l_m}$.
	Thus, from (\ref{ineq:peniultimum_thm}) and (\ref{eq:p_cimax_min_p_ci}), we can see that $\mathds{P}\left[C_{i_{max}}\right] > \mathds{P}\left[C_{i}\right]$ holds.
	
	\textit{Subcase 2.2b:} Finally, suppose that $\boldsymbol{x_{max}}[l_m] = 1$, i.e., $\boldsymbol{x_{max}} \in C^{l_m, 1}_{i_{max}}$.
	This implies $0.5 \leq p_{l_m} \leq 1$ and $p_{l_m} \geq \left(1 - p_{l_m}\right)$.
	By multiplying both sides of (\ref{ineq:second_ind_hyp_x_not_eq}) by $p_{l_m}$ and rearranging the terms, we have
	\begin{multline}\label{ineq:ultimum_thm}
		p_{l_m}\left(\mathds{P}\left[C^{\overline{l_m},1}_{i_{max}}\right] - \mathds{P}\left[C^{\overline{l_m},1}_{i}\right]\right) \geq \\ p_{l_m}\left(\mathds{P}\left[C^{\overline{l_m},0}_{i}\right] - \mathds{P}\left[C^{\overline{l_m},0}_{i_{max}}\right]\right) \geq \\
		\left(1 -p_{l_m}\right)\left(\mathds{P}\left[C^{\overline{l_m},0}_{i}\right] - \mathds{P}\left[C^{\overline{l_m},0}_{i_{max}}\right]\right),
	\end{multline}
	where the last inequality comes from  $p_{l_m} \geq \left(1 - p_{l_m}\right)$.
	By again rearranging the terms in the first and the last line of the inequality (\ref{ineq:ultimum_thm}), we get the inequality
	\begin{multline}\label{ineq:ultimum_rear_thm}
		\left(1 -p_{l_m}\right)\left(\mathds{P}\left[C^{\overline{l_m},0}_{i_{max}}\right] - \mathds{P}\left[C^{\overline{l_m},0}_{i}\right]\right)\\ \geq p_{l_m}\left(\mathds{P}\left[C^{\overline{l_m},1}_{i}\right] - \mathds{P}\left[C^{\overline{l_m},1}_{i_{max}}\right]\right) .
	\end{multline}
	From (\ref{eq:p_cimax_min_p_ci}) and (\ref{ineq:ultimum_rear_thm}), it directly follows $\mathds{P}\left[C_{i_{max}}\right] \geq \mathds{P}\left[C_{i}\right]$.
	
	By the principle of double induction, the theorem is true for all binary linear codes of any dimension $k$, $0 \leq k < n$, and all Hamming weights of their coset leaders HW$\left(\boldsymbol{e_{l}}\right) \geq 1$.
	
\end{proof}

The results of the coset-coset inequality theorem will be helpful in determining the exact output min-entropy of the linear corrector when the distributions of all raw input bits are precisely known.
For most real-world TRNGs, these distributions are unknown during the design time and vary, in some range, between TRNG instances and during the operation.
Often, the only thing that can be guaranteed and required by the standardization bodies \cite{Turan2018, AIS31,AIS31New} is the lower bound on entropy. 
Hence, to practically apply the finding of Theorem~\ref{thm:generalized_Sullivan}, that the most probable coset is the one that contains the most probable vector, we will use it in the following lemma to show how this probability can be bounded.

\begin{lemma}
	\label{lemma:helper}
	Let $C_0\left(p_0,\dots,p_{n-1}\right)$ be the set of codewords of a binary linear code and $C_{i_{max}}\left(p_0,\dots,p_{n-1}\right)$ be the most probable set as defined in Theorem~\ref{thm:generalized_Sullivan} with corresponding coordinate 1-probabilities given by tuple $\left(p_0,\dots,p_{n-1}\right)$, where all $p_i$ might be different.
Let $\delta_{max} = \text{max}\left\{\left|0.5 - p_i\right|\right\}_{i=0}^{n-1}$ be the maximum coordinate bit bias, and let $\left(0.5 - \delta_{max},\dots,0.5 - \delta_{max}\right)$ represent a tuple of coordinate 1-probabilities all equal to $0.5 - \delta_{max}$.
	Then, it holds $\mathds{P}\left[C_{0}\left(0.5 - \delta_{max},\dots,0.5 - \delta_{max}\right)\right] \geq \mathds{P}\left[C_{i_{max}}\left(p_0,\dots,p_{n-1}\right)\right]$.
\end{lemma}

\begin{proof}
	We will decompose the proof into two cases, depending on whether the most probable vector $\boldsymbol{x_{max}}$ is an all-zero vector, and prove both cases by simple induction.
	
	\textit{Case 1:} Suppose that the most probable vector $\boldsymbol{x_{max}}$ is the all-zero vector, i.e., all coordinate 1-probabilities $p_i$, $0 \leq i \leq n-1$, are lower than 0.5 and possibly different from each other. 
	According to Theorem~\ref{thm:generalized_Sullivan}, the most probable set will be $C_0$, i.e., $C_{i_{max}} = C_0$.
	If in some coordinate $l_0$, we change its probability $p_{l_0}$ to $p^*_{l_0} = 0.5 - \delta_{max}$, the all-zero vector will remain the most probable vector for the tuple of 1-probabilities $\left(p_0, \dots, p^*_{l_0}, \dots, p_{n-1} \right)$ and therefore $C_0$ remains the most probable set. 
We partition $C_0$ into two subsets $C^{l_0,\, b}_0 = \left\{\boldsymbol{x} \in C_{0} \mid \boldsymbol{x}[l_0] = b\right\}$, $b \in \{0,1\}$, according to the value of the element vectors' coordinate in $l_0$. 
We now remove the $l_0$ coordinate of each element in $C^{l_0,\, b}_0$ and obtain subsets $C^{\overline{l_0},\, b}_0$, $b \in \{0,1\}$.
	Note that $\mathds{P}\left[C^{\overline{l_0},\, b}_0 \left(p_0, \dots, p_{l_0}, \dots, p_{n-1}\right)\right] = \mathds{P}\left[C^{\overline{l_0},\, b}_0 \left(p_0, \dots, p^*_{l_0}, \dots, p_{n-1}\right)\right] = \mathds{P}\left[C^{\overline{l_0},\, b}_0\right]$, since the vectors in $C^{\overline{l_0},\, b}_0$ do not have coordinate $l_0$ with modified probability.
	The probability of $C_0$ before and after the $l_0$ coordinate probability change will be
	\begin{multline}\label{eq:p_c_0_orig}
		\mathds{P}\left[C_0 \left(p_0, \dots, p_{l_0}, \dots, p_{n-1}\right)\right] = \\
		p_{l_0}\mathds{P}\left[C^{\overline{l_0},1}_0\right] + \left(1 -p_{l_0}\right)\mathds{P}\left[C^{\overline{l_0},0}_0\right], 
	\end{multline}
	and
	\begin{multline}\label{eq:p_c_0_mod}
		\mathds{P}\left[C_0 \left(p_0, \dots, p^*_{l_0} = 0.5-\delta_{max}, \dots, p_{n-1}\right)\right] = \\ p^*_{l_0}\mathds{P}\left[C^{\overline{l_0},1}_0\right] + \left(1 -p^*_{l_0}\right)\mathds{P}\left[C^{\overline{l_0},0}_0\right] = \\ \left(0.5 - \delta_{max}\right)\mathds{P}\left[C^{\overline{l_0},1}_0\right] + \left(0.5 + \delta_{max}\right)\mathds{P}\left[C^{\overline{l_0},0}_0\right],
	\end{multline}
	respectively.
	By subtracting (\ref{eq:p_c_0_orig}) from (\ref{eq:p_c_0_mod}), we obtain
	\begin{multline}\label{eq:c_0_sub}
		\mathds{P}\left[C_0 \left(p_0, \dots, p^*_{l_0} = 0.5-\delta_{max}, \dots, p_{n-1}\right)\right] \\ - \mathds{P}\left[C_0 \left(p_0, \dots, p_{l_0}, \dots, p_{n-1}\right)\right] = \\ \left(0.5 - \delta_{max} - p_{l_0}\right)\mathds{P}\left[C^{\overline{l_0},1}_0\right] +  \left(p_{l_0} - 0.5 + \delta_{max}\right)\mathds{P}\left[C^{\overline{l_0},0}_0\right]\\ = \left(\delta_{max} - \left(0.5 - p_{l_0}\right)\right)\left(\mathds{P}\left[C^{\overline{l_0},0}_0\right] - \mathds{P}\left[C^{\overline{l_0},1}_0\right]\right).
	\end{multline}
	The first multiplication term in the last line of (\ref{eq:c_0_sub}) is non-negative since, by the definition of $\delta_{max}$, it holds $\delta_{max} \geq 0.5 - p_{l_0}$.

	\textit{Subcase 1.1:} If for all $\boldsymbol{x} \in C_0$ it holds $\boldsymbol{x}[l_0] = 0$, then $C^{l_0,0}_0 = C_0$ and $C^{l_0,1}_0 = C^{\overline{l_0},1}_0 = \varnothing$.
	This implies that the second multiplication term in (\ref{eq:c_0_sub}) is also non-negative, since $\mathds{P}\left[C^{\overline{l_0},1}_0\right] = 0$ and therefore $\mathds{P}\left[C_0 \left(p_0, \dots, p^*_{l_0} = 0.5-\delta_{max}, \dots, p_{n-1}\right)\right] \geq \mathds{P}\left[C_0 \left(p_0, \dots, p_{l_0}, \dots, p_{n-1}\right)\right]$.
	
	\textit{Subcase 1.2:} If $\boldsymbol{x}[l_0]$ is not identical for all $\boldsymbol{x} \in C_0$, we have two additional subcases depending on whether $C_0$ contains the vector element $\boldsymbol{1_{l_0}}$ -- an $n$-bit vector with Hamming weight 1 that has a 1 in coordinate $l_0$.
	
	\textit{Subsubcase 1.2.1:} If $\boldsymbol{1_{l_0}} \in C_0$, then
	every element $\boldsymbol{c}$ in $C^{l_0, 0}_{0}$ has exactly one corresponding element $\boldsymbol{c^{\prime}}$ in $C^{l_0,\, 1}_{0}$ to which it is related by $\boldsymbol{c^{\prime}} = \boldsymbol{c} \oplus \boldsymbol{1_{l_0}}$. 
	Then $C^{\overline{l_0},1}_0 = C^{\overline{l_0},0}_0$, and the second multiplication term in (\ref{eq:c_0_sub}) is 0.
	Hence, $\mathds{P}\left[C_0 \left(p_0, \dots, p^*_{l_0} = 0.5-\delta_{max}, \dots, p_{n-1}\right)\right] = \mathds{P}\left[C_0 \left(p_0, \dots, p_{l_0}, \dots, p_{n-1}\right)\right]$.

    \textit{Subsubcase 1.2.2:} Suppose that $\boldsymbol{1_{l_0}} \notin C_0$.
	Then by Theorem~\ref{thm:generalized_Sullivan}, $\mathds{P}\left[C^{\overline{l_0},0}_0\right] \geq \mathds{P}\left[C^{\overline{l_0},1}_0\right]$, since the set $C^{\overline{l_0},0}_0$ contains the ${\left(n-1\right)}$-bit all-zero vector and $C^{\overline{l_0},1}_0$ is its proper coset.
	Hence, the second multiplication term in the last line of (\ref{eq:c_0_sub}) is also non-negative and $\mathds{P}\left[C_0 \left(p_0, \dots, p^*_{l_0} = 0.5-\delta_{max}, \dots, p_{n-1}\right)\right] \geq \mathds{P}\left[C_0 \left(p_0, \dots, p_{l_0}, \dots, p_{n-1}\right)\right]$.

	By a trivial induction over coordinates $l_i$, $0 \leq i \leq n-1$, and iteratively applying the described coordinate probability substitution, one can easily arrive at the lemma's inequality for Case 1:
	\begin{equation}\label{eq:x_max_all_zero_main_ineq}
		\mathds{P}\left[C_0 \left(0.5-\delta_{max}, \dots, 0.5-\delta_{max}\right)\right] \geq \mathds{P}\left[C_0 \left(p_0, \dots, p_{n-1}\right)\right].
	\end{equation} 
	
	\textit{Case 2:} Suppose that the most probable vector $\boldsymbol{x_{max}}$ is not the all-zero vector, i.e., HW$\left(\boldsymbol{x_{max}}\right) = r \geq 1$ with 1-probabilities not smaller than 0.5 in coordinates ${l_0,\dots,l_{r-1}}$.
	If we change one of the coordinate 1-probabilities $p_{l_0}$, that was not smaller than 0.5, to $p^*_{l_0} = 0.5 - \delta_{max}$, the new most probable vector $\boldsymbol{x^{l_0, 0}_{max}}$ will be equal to $\boldsymbol{x_{max}}$ in all coordinates except $l_0$, in which $\boldsymbol{x^{l_0, 0}_{max}}$ has a 0.
	
	\textit{Subcase 2.1:} If $\boldsymbol{1_{l_0}} \in C_0$, then $\boldsymbol{x_{max}}$ and $\boldsymbol{x^{l_0, 0}_{max}}$ are in the same set $C_{i_{max}}$ since $\boldsymbol{x_{max}} = \boldsymbol{x^{l_0, 0}_{max}} \oplus \boldsymbol{1_{l_0}}$.
	All vectors in $C_{i_{max}}$ can be divided into two subsets $C^{l_0,\, b}_{i_{max}} = \left\{\boldsymbol{x} \in C_{i_{max}} \mid \boldsymbol{x}[l_0] = b\right\}$, $b \in \{0,1\}$.
	We now remove the coordinate $l_0$ of each element in $C^{l_0,\, b}_{i_{max}}$ to obtain subsets $C^{\overline{l_0},\, b}_{i_{max}}$, $b \in \{0,1\}$.
	Since every element $\boldsymbol{c_{i_{max}}}$ in $C^{l_0,\, b}_{i_{max}}$ has exactly one corresponding element $\boldsymbol{c^{\prime}_{i_{max}}}$ in $C^{l_0,1-b}_{i_{max}}$ to which it is related by $\boldsymbol{c^{\prime}_{i_{max}}} = \boldsymbol{c_{i_{max}}} \oplus  \boldsymbol{1_{l_0}}$, it is clear that $C^{\overline{l_0},\, b}_{i_{max}} = C^{\overline{l_0},1-b}_{i_{max}}$.
Then, for the probabilities of $C_{i_{max}}\left(p_0, \dots, p_{l_0}, \dots, p_{n-1}\right)$ and $C_{i_{max}}\left(p_0, \dots, p^*_{l_0}, \dots, p_{n-1}\right)$, we have
	\begin{multline}
		\mathds{P}\left[C_{i_{max}}\left(p_0, \dots, p_{l_0}, \dots, p_{n-1}\right)\right] = p_{l_0}\mathds{P}\left[C^{\overline{l_0},1}_{i_{max}}\right] \\  + \left(1 -p_{l_0}\right)\mathds{P}\left[C^{\overline{l_0},0}_{i_{max}}\right] = \mathds{P}\left[C^{\overline{l_0},0}_{i_{max}}\right] = \mathds{P}\left[C^{\overline{l_0},1}_{i_{max}}\right],
	\end{multline}
	and
	\begin{multline}
		\mathds{P}\left[C_{i_{max}}\left(p_0, \dots, p^*_{l_0}, \dots, p_{n-1}\right)\right] = p^*_{l_0}\mathds{P}\left[C^{\overline{l_0},1}_{i_{max}}\right] \\ + \left(1 -p^*_{l_0}\right)\mathds{P}\left[C^{\overline{l_0},0}_{i_{max}}\right] = \mathds{P}\left[C^{\overline{l_0},0}_{i_{max}}\right] = \mathds{P}\left[C^{\overline{l_0},1}_{i_{max}}\right],
	\end{multline} 
	respectively.
	Therefore, $\mathds{P}\left[C_{i_{max}}\left(p_0, \dots, p_{l_0}, \dots, p_{n-1}\right)\right] = \mathds{P}\left[C_{i_{max}}\left(p_0,\dots, p^*_{l_0} = 0.5-\delta_{max}, \dots,p_{n-1}\right)\right]$.
	
	\textit{Subcase 2.2:} If $\boldsymbol{1_{l_0}} \notin C_0$, then $\boldsymbol{x_{max}}$ and $\boldsymbol{x^{l_0, 0}_{max}}$ will be in different sets, which we denote by $C_{i_{max}}$ and $C_{i_{max}, l_0}$.
	Let $\boldsymbol{e_{x_{max}}}$ be a vector element of the lowest weight in $C_{i_{max}}$ and let $\boldsymbol{c_{0,x_{max}}}$ be the codeword such that $\boldsymbol{x_{max}} = \boldsymbol{e_{x_{max}}} \oplus \boldsymbol{c_{0,x_{max}}}$.
	Similarly, let $\boldsymbol{e_{x_{max}, l_0}}$ be a vector element of the lowest weight in $C_{i_{max}, l_0}$ and let $\boldsymbol{c_{0,x_{max},l_0}}$ be the codeword such that $\boldsymbol{x^{l_0, 0}_{max}} = \boldsymbol{e_{x_{max}, l_0}} \oplus \boldsymbol{c_{0,x_{max},l_0}}$.
	Since $\boldsymbol{x_{max}} = \boldsymbol{x^{l_0, 0}_{max}} \oplus \boldsymbol{1_{l_0}}$, it holds
	\begin{equation}\label{eq:1_l_0}
		\boldsymbol{1_{l_0}} = \boldsymbol{c_{0,x_{max}}} \oplus \boldsymbol{c_{0,x_{max},l_0}} \oplus \boldsymbol{e_{x_{max}}} \oplus \boldsymbol{e_{x_{max},l_0}}.
	\end{equation}
	Every element $\boldsymbol{e_{x_{max}}} \oplus \boldsymbol{c_{0,j}}$ from $C_{i_{max}}$ has one corresponding element in the set $C_{i_{max}, l_0}$ from which it differs only in the coordinate $l_0$:
	\begin{multline}\label{eq:rel_x_max_x_max_m_1_different}
		\boldsymbol{e_{x_{max}}} \oplus \boldsymbol{c_{0,j}} \oplus \boldsymbol{1_{l_0}} = \\ \boldsymbol{e_{x_{max}, l_0}} \oplus \boldsymbol{c_{0,j}} \oplus \boldsymbol{c_{0,x_{max}}} \oplus \boldsymbol{c_{0,x_{max}, l_0}}.
	\end{multline}
We partition $C_{i_{max}}$ into subsets $C^{l_0,\, b}_{i_{max}} = \left\{\boldsymbol{x} \in C_{i_{max}} \mid \boldsymbol{x}[l_0] = b\right\}$ and $C_{i_{max},\, l_0}$ into subsets $C^{l_0,\, b}_{i_{max},l_0} = \left\{\boldsymbol{x} \in C_{i_{max},\, l_0} \mid \boldsymbol{x}[l_0] = b\right\}$, $b \in \left\{0,1\right\}$, according to the value of the coordinate $l_0$.
	We also remove components in coordinate $l_0$ of each element in $C^{l_0,\, b}_{i_{max}}$ and $C^{l_0,\, b}_{i_{max},l_0}$, $b \in \left\{0,1\right\}$, and obtain $C^{\overline{l_0},\, b}_{i_{max}}$ and $C^{\overline{l_0},\, b}_{i_{max},l_0}$ , respectively.
	Due to (\ref{eq:rel_x_max_x_max_m_1_different}), the elements in $C_{i_{max}}$ and $C_{i_{max},\, l_0}$ differ only in the $l_0$ coordinate, and it follows $C^{\overline{l_0},\, b}_{i_{max}} = C^{\overline{l_0},1-b}_{i_{max},l_0}$.
	Then for the probabilities of $C_{i_{max}}\left(p_0, \dots, p_{l_0}, \dots, p_{n-1}\right)$ and $C_{i_{max},\, l_0}\left(p_0, \dots, p^*_{l_0}, \dots, p_{n-1}\right)$, we have
\begin{multline}\label{eq:p_cmax_m}
		\mathds{P}\left[C_{i_{max}}\left(p_0,\dots, p_{l_0}, \dots,p_{n-1}\right)\right] = \\ p_{l_0}\mathds{P}\left[C^{\overline{l_0},1}_{i_{max}}\right] + \left(1 -p_{l_0}\right)\mathds{P}\left[C^{\overline{l_0},0}_{i_{max}}\right], 
	\end{multline} 
	and
	\begin{multline}\label{eq:p_cmax_m_1}
		\mathds{P}\left[C_{i_{max},\, l_0}\left(p_0,\dots, p^*_{l_0}, \dots,p_{n-1}\right)\right] = \\ p^*_{l_0}\mathds{P}\left[C^{\overline{l_0},1}_{i_{max},l_0}\right] + \left(1 -p^*_{l_0}\right)\mathds{P}\left[C^{\overline{l_0},0}_{i_{max},l_0}\right] = \\ \left(0.5 + \delta_{max}\right)\mathds{P}\left[C^{\overline{l_0},1}_{i_{max}}\right] + \left(0.5 - \delta_{max}\right)\mathds{P}\left[C^{\overline{l_0},0}_{i_{max}}\right],
	\end{multline} 
	respectively.
	By subtracting (\ref{eq:p_cmax_m}) from (\ref{eq:p_cmax_m_1}), we obtain
	\begin{multline}\label{eq:c_max_m_1_c_max_m}
		\mathds{P}\left[C_{i_{max},\, l_0}\left(p_0,\dots, p^*_{l_0} = 0.5-\delta_{max}, \dots,p_{n-1}\right)\right] \\ - \mathds{P}\left[C_{i_{max}}\left(p_0,\dots, p_{l_0}, \dots,p_{n-1}\right)\right] = \\ \left(0.5 + \delta_{max} - p_{l_0}\right)\mathds{P}\left[C^{\overline{l_0},1}_{i_{max}}\right] +  \left(p_{l_0} - 0.5 - \delta_{max}\right)\mathds{P}\left[C^{\overline{l_0},0}_{i_{max}}\right]\\ = \left(\delta_{max} - \left(p_{l_0} - 0.5\right)\right)\left(\mathds{P}\left[C^{\overline{l_0},1}_{i_{max}}\right] - \mathds{P}\left[C^{\overline{l_0},0}_{i_{max}}\right]\right).
	\end{multline}
By the definition of $\delta_{max}$, it follows $\delta_{max} \geq p_{l_0} - 0.5$, and thus, the first multiplication term in the last line of (\ref{eq:c_max_m_1_c_max_m}) is non-negative.
	Since in this case $\boldsymbol{x_{max}}$ has a 1 in coordinate $l_0$, we have two possibilities for the vectors in $C_{i_{max}}$: either $\boldsymbol{x}[l_0] = 1$ holds for all $\boldsymbol{x} \in C_{i_{max}}$ or half of the vectors have a 0 and the other half have a 1 in coordinate $l_0$. 
	
	\textit{Subsubcase 2.2.1:} If $\boldsymbol{x}[l_0]=1$ holds for all $\boldsymbol{x} \in C_{i_{max}}$, then $C^{l_0,1}_{i_{max}} = C_{i_{max}}$ and $C^{l_0,0}_{i_{max}} = C^{\overline{l_0},0}_{i_{max}} = \varnothing$.
	Since $\mathds{P}\left[C^{\overline{l_0},0}_{i_{max}}\right]=0$, the second multiplication term in the last line of (\ref{eq:c_max_m_1_c_max_m}) is also non-negative and thus, $\mathds{P}\left[C_{i_{max},\, l_0}\left(p_0,\dots, p^*_{l_0} = 0.5-\delta_{max}, \dots,p_{n-1}\right)\right] \geq \mathds{P}\left[C_{i_{max}}\left(p_0,\dots, p_{l_0}, \dots,p_{n-1}\right)\right]$.
	
	\textit{Subsubcase 2.2.2:} If not all $\boldsymbol{x} \in C_{i_{max}}$ have identical bit in position $l_0$, then the most probable ${\left(n-1\right)}$-bit vector obtained from $\boldsymbol{x_{max}}$ by deleting its $l_0$ coordinate $\boldsymbol{x^{\overline{l_0}}_{max}}$ will be in $C^{\overline{l_0},1}_{i_{max}}$.
	By Theorem~\ref{thm:generalized_Sullivan}, we have $\mathds{P}\left[C^{\overline{l_0},1}_{i_{max}}\right] \geq \mathds{P}\left[C^{\overline{l_0},0}_{i_{max}}\right]$.
	Thus, the second multiplication term in the last line of (\ref{eq:c_max_m_1_c_max_m}) is also non-negative, and we have $\mathds{P}\left[C_{i_{max},\, l_0}\left(p_0,\dots, p^*_{l_0} = 0.5-\delta_{max}, \dots,p_{n-1}\right)\right] \geq \mathds{P}\left[C_{i_{max}}\left(p_0,\dots, p_{l_0}, \dots,p_{n-1}\right)\right]$.

	Similarly to Case 1, we can use trivial induction over all coordinates $l_i$, $0 \leq i \leq r-1$, in which $\boldsymbol{x_{max}}$ has value 1.
	By iteratively applying the described coordinate probability substitutions, we are saddled with the most probable vector $\boldsymbol{x^{l_0, 0;\dots;l_{r-1},0}_{max}}$, which is an all-zero vector since all ones are replaced by zeros and therefore $C_{i_{max},\, l_0,\dots, l_{r-1}} = C_0$.
	Hence, we arrive at the inequality:
	\begin{multline}
		\mathds{P}\left[C_{0}\left(p_0,\dots, p^*_{l_0}, \dots, p^*_{l_{r-1}}, \dots, p_{n-1}\right)\right] = \\
		\mathds{P}\left[C_{i_{max},\, l_0,\dots, l_{r-1}}\left(p_0,\dots, p^*_{l_0}, \dots, p^*_{l_{r-1}}, \dots, p_{n-1}\right)\right] \\ \geq
		\mathds{P}\left[C_{i_{max}}\left(p_0,\dots, p_{l_0}, \dots, p_{l_{r-1}}, \dots, p_{n-1}\right)\right],
	\end{multline}
	where $p^*_{l_0} = \dots = p^*_{l_{r-1}} = 0.5 - \delta_{max}$.
We can now apply the inequality (\ref{eq:x_max_all_zero_main_ineq}) from \textit{Case 1} when the all-zero is the most probable vector to obtain the lemma's inequality:
	\begin{multline}
		\mathds{P}\left[C_0 \left(0.5-\delta_{max}, \dots, 0.5-\delta_{max}\right)\right] \geq \\ \mathds{P}\left[C_{0}\left(p_0,\dots, p^*_{l_0}, \dots, p^*_{l_{r-1}}, \dots, p_{n-1}\right)\right] \geq \\
		\mathds{P}\left[C_{i_{max}}\left(p_0,\dots, p_{l_0}, \dots, p_{l_{r-1}}, \dots, p_{n-1}\right)\right].
	\end{multline}
	Since we have shown that the lemma's inequality is satisfied for all possible Hamming weights of the $\boldsymbol{x_{max}}$, this concludes the proof.

\end{proof}

We use the previous results from Theorem \ref{thm:generalized_Sullivan} and Lemma \ref{lemma:helper} for the main theorem that provides a lower bound on the min-entropy of the output of a linear corrector when only a lower bound on the min-entropy of the noise source of independent bits is known.
	
\begin{theorem}\label{thm:main_thm}
	Let $\mathcal{X}$ be a row vector $n$-bit random variable with independent but not necessarily identically distributed coordinates and let the min-entropy per bit of $\mathcal{X}$ be at least $\mathrm{H}^{in}_{\infty} > 0$.
	Let $\boldsymbol{G}$ be a $k \times n$ generator matrix of a binary linear $\left[n, k\right]$-code $C_0$ and let $\left(\mathfrak{A}_i\right)_{i=0}^n$ be its weight distribution.
Then, the total min-entropy of the output of the linear corrector $\mathcal{Y} = \boldsymbol{G}  \mathcal{X}^{\intercal}$ is lower-bounded by:
	\begin{multline}\label{eq:Hmin_pp}
\mathrm{H}^{out,\,tot}_{\infty} \geq - \log_{2} \left(2^{-k}\sum_{i=0}^{n}\mathfrak{A}_i \left(2^{1 -\mathrm{H}^{in}_{\infty}} - 1\right)^i \right).
	\end{multline}
\end{theorem}
	
\begin{proof}
	The proof will be a straightforward application of Theorem~\ref{thm:generalized_Sullivan} and Lemma~\ref{lemma:helper}.
	Consider the min-entropy of the $i$-th bit of $\mathcal{X}$ with 1-probability $0 < p_i < 1$:
	\begin{multline}
		\mathrm{H}^{in,i}_{\infty} = -\log_{2}\left(\text{max}\left\{p_i,\, 1-p_i\right\}\right) =\\ -\log_{2}\left(0.5 + \left|0.5 - p_i\right|\right).
	\end{multline}
	If we also denote the maximal bit bias with $\delta_{max} = \text{max}\left\{\left|0.5 - p_i\right|\right\}_{i=0}^{n-1}$, then the lower bound on the min-entropy per bit of $\mathcal{X}$ is simply given by
	\begin{multline}\label{eq:hin_delta_max}
		\mathrm{H}^{in}_{\infty} = \text{min}\left\{\mathrm{H}^{in,i}_{\infty}\right\}^{n-1}_{i=0} =\\ -\log_{2}\left(0.5 + \text{max}\left\{\left|0.5 - p_i\right|\right\}^{n-1}_{i=0} \right) =\\ - \log_{2}\left(0.5 + \delta_{max}\right).
	\end{multline}
By definition of the linear corrector $\mathcal{Y} = \boldsymbol{G}  \mathcal{X}^{\intercal}$ and the fact that the generator matrix $\boldsymbol{G}$ of the $C_0$ code is equivalent to the parity-check matrix of its dual code $C^{\bot}_0$, every $k$-bit output $\boldsymbol{y}$ of the linear corrector will be a syndrome for the dual code $C^{\bot}_0$ of an $n$-bit vector $\boldsymbol{x}$ which is a realization  of $\mathcal{X}$.
	Since all $n$-bit vectors belonging to the same coset of $C^{\bot}_0$ have the same syndrome, determining the probability of each output is equivalent to determining the probability of the corresponding coset of $C^{\bot}_0$.	
By Theorem~\ref{thm:generalized_Sullivan}, the most probable output will correspond to the syndrome of the most probable input vector.
	From (\ref{eq:hin_delta_max}) and the Lemma~\ref{lemma:helper}, it holds
	\begin{multline}\label{ineq:p_c_0_p_c_imax_main_thm}
		\underset{y \in \mathcal{Y}}{\text{max}}\, \mathds{P}\left[\mathcal{Y} = y\right] = \mathds{P}\left[C^{\bot}_{i_{max}}\left(p_0, \dots, p_{n-1}\right)\right] \leq \\ \mathds{P}\left[C^{\bot}_0\left(p = 1 - 2^{-\mathrm{H}^{in}_{\infty}}, \dots, p = 1 - 2^{-\mathrm{H}^{in}_{\infty}}\right)\right],
	\end{multline}
where $C^{\bot}_{i_{max}}$ is a coset of $C^{\bot}_0$ that contains the most probable vector $\boldsymbol{x}$.
	Since the number of vectors of $C^{\bot}_0$ with Hamming weight $i$ is given by its weight distribution $\left(\mathfrak{A}^{\bot}_i\right)_{i=0}^n$, we can determine the lower bound of the total output min-entropy as
\begin{multline}\label{ineq:tot_out_min_ent_last_main_thm}
		\mathrm{H}^{out,\,tot}_{\infty} = - \log_{2}\left(\underset{y \in \mathcal{Y}}{\text{max}}\, \mathds{P}\left[\mathcal{Y} = y\right]\right)  \geq \\  - \log_{2} \left(\sum_{i=0}^{n}\mathfrak{A}^{\bot}_i\, p^{i} \left(1 - p\right)^{n-i} \right)\\ =   - \log_{2} \left(2^{-n \mathrm{H}^{in}_{\infty}}\sum_{i=0}^{n}\mathfrak{A}^{\bot}_i \left(2^{\mathrm{H}^{in}_{\infty}} - 1\right)^i \right).
	\end{multline}
	By substituting $2^{1 - \mathrm{H}^{in}_{\infty}} - 1$ for $z$ in the MacWilliams identity (\ref{eq:macwilliams_id}), we have
	\begin{equation}
  \sum_{i=0}^{n}\mathfrak{A}^{\bot}_i \left(2^{\mathrm{H}^{in}_{\infty}} - 1\right)^i = 2^{n \mathrm{H}^{in}_{\infty}} 2^{-k}\sum_{i=0}^{n}\mathfrak{A}_i \left(2^{1 -\mathrm{H}^{in}_{\infty}} - 1\right)^i,
	\end{equation}
	and thus the theorem follows.

\end{proof}
	
According to Lemma~\ref{lemma:helper}, our new bound (\ref{eq:Hmin_pp}) is tight when independent input bits are not identically distributed and only the lower bound on the input min-entropy is known, and it is met with equality when independent input bits are identically distributed with $p < 0.5$.
In addition, thanks to Theorem~\ref{thm:generalized_Sullivan}, it is possible to determine the value and probability of the linear corrector's most probable output when the distributions of the input bits are precisely known.
Since $\sum_{i=d}^{n}\mathfrak{A}_i = 2^{k} - 1$ and $\left(2^{1 - \mathrm{H}^{in}_{\infty}} - 1\right)^i \leq \left(2^{1 - \mathrm{H}^{in}_{\infty}} - 1\right)^d$ for $i \geq d$, it is straightforward to show that the lower bound from Theorem~\ref{thm:main_thm} is always tighter than the overly conservative state-of-the-art bound given by (\ref{eq:Lacharme_ineq_mod}):
\begin{multline}\label{eq:tight_proof}
    - \log_{2} \left(2^{-k}\sum_{i=0}^{n}\mathfrak{A}_i \left(2^{1 - \mathrm{H}^{in}_{\infty}} - 1\right)^i\right) =\\ - \log_{2} \left(2^{-k} + 2^{-k}\sum_{i=d}^{n}\mathfrak{A}_i \left(2^{1 - \mathrm{H}^{in}_{\infty}} - 1\right)^i\right) \geq \\ - \log_{2} \left(2^{-k} + 2^{-k} \left(2^{k} - 1\right) \left(2^{1 - \mathrm{H}^{in}_{\infty}} - 1\right)^d\right) > \\- \log_{2} \left(2^{-k} + \left(2^{1 - \mathrm{H}^{in}_{\infty}} - 1\right)^d\right).
\end{multline}

Finally, it is worth mentioning, as pointed out by one of the reviewers, that the results presented in this section can alternatively be obtained using established Fourier techniques outlined in the works of Redinbo \cite{Redinbo1973} and Meneghetti \cite{ThesisMeneghetti2017}.

 \section{Selection of the Linear Correctors}
\label{Sec:Optimal_Selection}

Improvement of the new bound over the old one given by (\ref{eq:Lacharme_ineq_mod}) varies depending on the corrector's underlying code for which the bounds are calculated.
From (\ref{eq:Lacharme_ineq_mod}), it can be observed that for identical $\mathrm{H}^{in}_{\infty}$ and fixed corrector length $n$ and dimension $k$, the total output entropy is largest for the corrector based on a code with the greatest possible minimum distance $d$.
Linear codes that achieve the greatest minimum distance among all known $\left[n,k\right]$-codes are called the \textit{best known linear codes} (BKLCs) \cite{MAGMA, Grassl_codetables}.
On the other hand, it is clear from (\ref{eq:Hmin_pp}) that the relationship between $\mathrm{H}^{in}_{\infty}$ and $\mathrm{H}^{out,\, tot}_\infty$ is more complex and the codes' complete weight distribution should be considered. 
However, computing the weight distribution of a general binary linear code is an NP-hard problem \cite{Berlekamp78} and requires a significant computing effort for codes with high dimensions and high differences between the length and dimension.
In this section, we first calculate the new bound for the correctors based on the codes from the set of linear codes whose weight distributions can be conveniently determined or already available in the literature.
We then outline the process of selecting the optimal corrector for a given min-entropy rate of raw bits that maximizes the throughput of post-processed bits while maintaining the desired security level.
To demonstrate the practical advantages of our new bound, we compare the efficiencies and output min-entropies of correctors selected using the new bound against those selected using the old one.

\subsection{Optimal Extracting Linear Correctors}
\label{subsec:Optimal_Extracting_Correctors}

Both large output min-entropy and low throughput reduction are desirable corrector's properties.
Most security applications and standards \cite{Turan2018, AIS31, AIS31New} specify the output entropy requirements in terms of the min-entropy per bit $\mathrm{H}^{out,\, 1}_{\infty}$. 
To conservatively guarantee the entropy rate $\mathrm{H}^{out,\, 1}_{\infty}$ for every output bit, we require the total output min-entropy to be at least $\mathrm{H}^{out,\, tot}_\infty = k - 1 + \mathrm{H}^{out,\, 1}_\infty$. 
This requirement is more strict than $\mathrm{H}^{out,\, tot}_\infty = k \mathrm{H}^{out,\, 1}_\infty$, which would only guarantee the average min-entropy rate $\mathrm{H}^{out,\, 1}_{\infty}$ across all output bits, while the min-entropy of individual bits might be lower.
Since the throughput reduction is equal to the inverse of the underlying code's rate, a corrector based on a linear code is \textit{optimal extracting} if there are no codes in the considered set with simultaneously higher code rate $\nicefrac{k}{n}$ and a lower or equal required $\mathrm{H}^{in}_\infty$ to achieve $\mathrm{H}^{out,\, 1}_{\infty}$.
We denote this value of $\mathrm{H}^{in}_\infty$ as $\mathrm{H}^{in,\, req}_\infty$.
Post-processing of the raw bits with some specific (targeted) min-entropy rate is performed by selecting an optimal extracting corrector whose $\mathrm{H}^{in,\, req}_\infty$ is closest to the targeted min-entropy from below. 
By doing so, $\mathrm{H}^{out,\, 1}_\infty$ can be obtained at the corrector's output with the lowest possible throughput reduction.

\begin{figure}[!t]
	\centering
	\includegraphics{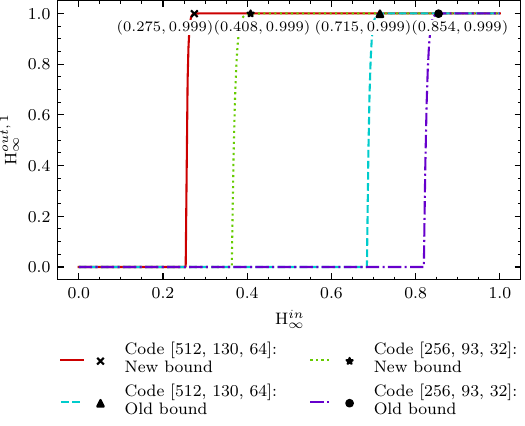}
	\caption{Relation between input and output min-entropy rate according to both old and new bounds for Reed-Muller $\left[512, 130, 64\right]$ and $\left[256, 93, 32\right]$ code-based correctors. The output min-entropy rate is computed as $\mathrm{H}^{out,\, 1}_\infty = \mathrm{max} \left( \mathrm{H}^{out,\, tot}_\infty - k + 1,\, 0 \right)$, where $\mathrm{H}^{out,\, tot}_\infty = f\left(\mathrm{H}^{in}_{\infty}\right)$ is determined for both the old and the new bound. All min-entropy values are rounded to three decimals.}
	\label{fig:Max_entropy_diff}
\end{figure}

\begin{figure*}[t]
	\centering
	\subfloat[\hspace*{-4em} \label{fig:a_Optimal_selection_fig_0.999}]{\includegraphics{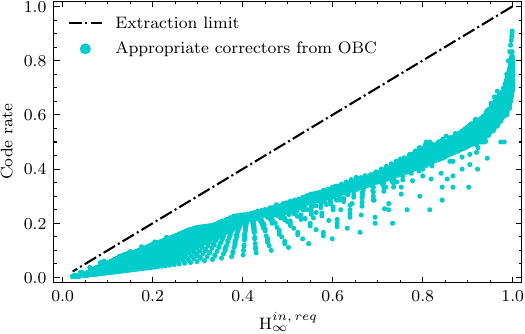}}
	\hfill
	\subfloat[\hspace*{-4em} \label{fig:b_Optimal_selection_fig_0.999}]{\includegraphics{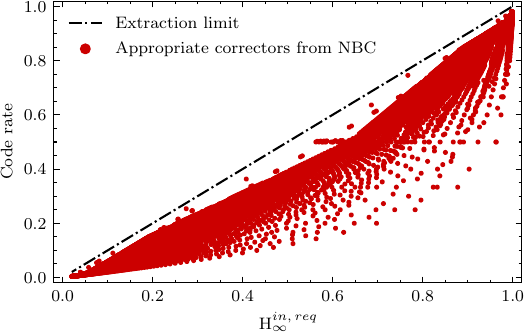}}
	\hfill
	\subfloat[\hspace*{-4em} \label{fig:c_Optimal_selection_fig_0.999}]{\includegraphics{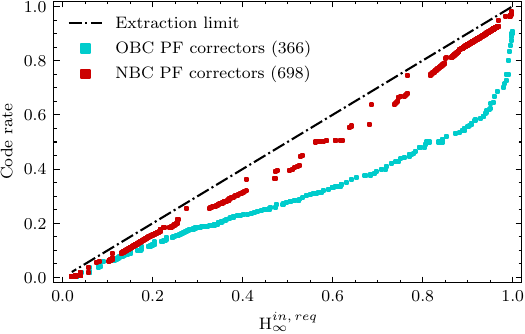}}
	\hfill
	\subfloat[\hspace*{-4em} \label{fig:d_Optimal_selection_fig_0.999}]{\includegraphics{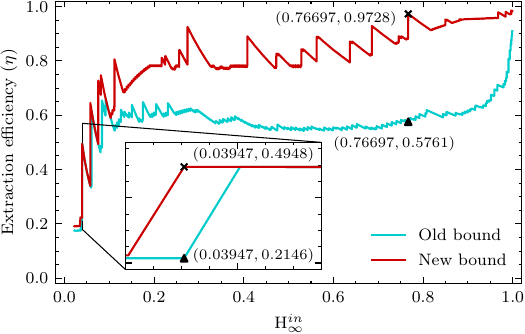}}
	\caption{Performances of linear correctors from OBC and NBC for $\mathrm{H}^{out,\, 1}_\infty \geq 0.999$ and extraction efficiency according to the old and the new bound.} 
	\label{fig:Optimal_selection_fig_0.999}
\end{figure*}

\begin{table}[t!]
	\caption{Optimal Linear Correctors and Performances for $\mathrm{H}^{out,\, 1}_{\infty} \geq 0.999$}
	\label{tab:PF_codes_0.999}
	\setlength\tabcolsep{2pt}
	\setlength\extrarowheight{2pt} 
	\begin{threeparttable}
		\begin{tabularx}{\linewidth}{@{\extracolsep{4pt}} *{5}{C} @{}}
			\toprule
			\multirow{3}{*}{Target $\mathrm{H}^{in}_{\infty}$} 
			& \multicolumn{2}{c}{Corrector construction} & \multicolumn{2}{c}{Extraction efficiency ($\eta$)} \\ 
			\cline{2-3} \cline{4-5} & OBC & NBC & Old bound & New bound\\
			\midrule
			\multirow{2}{*}{$0.1$}   & \multirow{2}{*}{$\left[511, 31, 219\right]$\tnote{a}}    & \multirow{2}{*}{$\left[511, 31, 219\right]$\tnote{a}}  & \multirow{2}{*}{$0.60665234$}    & $0.60665360$ $\left(+0.0002\% \right)$ \\ 
			\multirow{2}{*}{$0.2$}   & \multirow{2}{*}{$\left[254, 31, 96\right]$\tnote{*}}    & \multirow{2}{*}{$\left[243, 38, 83\right]$\tnote{*}}  & \multirow{2}{*}{$0.61022$}    & $0.78187$ $\left(+28.13\% \right)$ \\ 
			\multirow{2}{*}{$0.3$}   & \multirow{2}{*}{$\left[255, 47, 85\right]$\tnote{a,b}}    & \multirow{2}{*}{$\left[512, 130, 64\right]$\tnote{c}}  & \multirow{2}{*}{$0.61437$}    & $0.84635$ $\left(+37.76\% \right)$ \\ 
			\multirow{2}{*}{$0.4$}   & \multirow{2}{*}{$\left[126, 29, 42\right]$\tnote{*}}    & \multirow{2}{*}{$\left[122, 38, 31\right]$\tnote{*}}  & \multirow{2}{*}{$0.57538$}    & $0.77867$ $\left(+35.33\% \right)$ \\ 
			\multirow{2}{*}{$0.5$}   & \multirow{2}{*}{$\left[127, 35, 36\right]$\tnote{*}}    & \multirow{2}{*}{$\left[127, 50, 27\right]$\tnote{a,b}}  & \multirow{2}{*}{$0.55117$}    & $0.78740$ $\left(+42.86\% \right)$ \\ 
			\multirow{2}{*}{$0.6$}  & \multirow{2}{*}{$\left[87, 29, 24\right]$\tnote{*}}    & \multirow{2}{*}{$\left[127, 64, 19\right]$\tnote{d}}  & \multirow{2}{*}{$0.55554$}    & $0.83989$ $\left(+51.19\% \right)$ \\ 
			\multirow{2}{*}{$0.7$}   & \multirow{2}{*}{$\left[59, 23, 16\right]$\tnote{*}}    & \multirow{2}{*}{$\left[256, 163, 16\right]$\tnote{c}}  & \multirow{2}{*}{$0.55688$}    & $0.90960$ $\left(+63.34\% \right)$ \\ 
			\multirow{2}{*}{$0.8$}   & \multirow{2}{*}{$\left[46, 22, 12\right]$\tnote{*}}    & \multirow{2}{*}{$\left[512, 382, 16\right]$\tnote{c}}  & \multirow{2}{*}{$0.59781$}    & $0.93262$ $\left(+56.01\% \right)$ \\ 
			\multirow{2}{*}{$0.9$}   & \multirow{2}{*}{$\left[63, 35, 12\right]$\tnote{*}}    & \multirow{2}{*}{$\left[255, 219, 10\right]$\tnote{*}}  & \multirow{2}{*}{$0.61727$}    & $0.95424$ $\left(+54.59\% \right)$ \\ 
			\bottomrule
		\end{tabularx}
		\begin{tablenotes}[para, flushleft]\footnotesize
\item[*] BKLC code from \cite{Grassl_codetables}
			\item[a] BCH code from \cite{Lin2004}
\item[b] BCH code from \cite{Teradaweb}
			\item[c] Reed-Muller code from \cite{Sugita96, Teradaweb}
			\item[d] Quadratic residue code from \cite{oeis}
		\end{tablenotes}
	\end{threeparttable}
\end{table}

\subsection{Construction of Corrector Sets}
\label{subsec:Construction of Corrector Sets}

We first construct two sets from which the optimal extracting correctors will be determined: the set of correctors with output min-entropy determined according to the old bound (OBC) and the set of correctors with output min-entropy calculated by the new bound from Theorem~\ref{thm:main_thm} (NBC).
The OBC is a set of $32,741$ elements and consists of the correctors based on the non-trivial ($n \neq k$) BKLCs from \cite{Grassl_codetables}, BCH codes up to length 511 from \cite{Lin2004} and binary linear codes available at \cite{oeis, Truong05, Tomlinson17, Teradaweb, Schomaker92, Desaki97, Sugita96, Fujiwara93}.
On the other hand, the NBC set has a total of $16,613$ elements. 
It comprises correctors that are derived from binary linear codes with known weight distributions. These weight distributions are obtained from various sources, namely \cite{oeis, Truong05, Tomlinson17, Teradaweb, Schomaker92, Desaki97, Sugita96, Fujiwara93}. 
Additionally, the NBC set includes all non-trivial BKLCs and BCH codes found in the OBC. 
The length of these codes is restricted to $n < 81$, except for those with $n \geq 81$ that satisfy the condition $\mathrm{min}\left(k,\, n-k\right) \leq 38$.
Computing the weight distributions using MAGMA \cite{MAGMA} of BKLCs and BCH codes under these restrictions requires at most 60s per code of the real CPU time on Intel(R) Xeon(R) Gold 6248R CPU @ 3.00GHz with 24 cores and 48 threads.
To handle the BKLCs with generator matrices that contain one or more all-zero columns, we used codes with equivalent minimum distances but modified generator matrices to ensure that each column had at least one non-zero entry.

For hardware implementations of the correctors, opting for cyclic codes generally results in smaller area requirements. 
This is because they can be implemented with only several registers and XOR gates, utilizing the well-known generator or parity-check polynomial constructions \cite{Kwok2011, Lacharme2008}.
To also provide optimal extracting correctors based only on cyclic codes, we form two new sets out of OBC and NBC, consisting only of cyclic constructions -- OBCCYC and NBCCYC.
The OBCCYC and NBCCYC sets consist out of 803 and 637 correctors, respectively.
Comprehensive lists of elements in all four sets, accompanied by corresponding weight distributions for the NBC and NBCCYC sets, are publicly available via our Github repository \cite{OurGithub}.

Once the design parameter $\mathrm{H}^{out,\, 1}_\infty$ has been set, we calculate the code rate $\nicefrac{k}{n}$ and $\mathrm{H}^{in,\, req}_\infty$ according to (\ref{eq:Lacharme_ineq_mod}) for each corrector in the OBC and OBCCYC sets such that $\mathrm{H}^{out,\, tot}_\infty = k - 1 + \mathrm{H}^{out,\, 1}_\infty$ is reached.
Likewise, by numerically solving (\ref{eq:Hmin_pp}) via bisection for the same $\mathrm{H}^{out,\, tot}_\infty$, we obtain $\mathrm{H}^{in,\, req}_\infty$ and the code rate for every corrector in the NBC and NBCCYC sets.
If $\mathrm{H}^{in,\, req}_\infty$ is smaller than $\mathrm{H}^{out,\, 1}_\infty$, the corrector can be used for increasing the min-entropy rate and is referred to as an \textit{appropriate corrector}.
We form the subsets of appropriate correctors from each of the four corrector sets.
Finally, we construct sets of \textit{optimal extracting} correctors from sets of appropriate correctors, which we also call Pareto frontier (PF) correctors. 
It is important to note that, due to the disparity between the new and old bound, the optimal extracting correctors within the NBC/NBCCYC sets generally do not correspond to the optimal extracting correctors within the OBC/OBCYC sets.

\begin{figure*}[t]
	\centering
	\subfloat[\hspace*{-4em} \label{fig:a_CYC_Optimal_selection_fig_0.999}]{\includegraphics{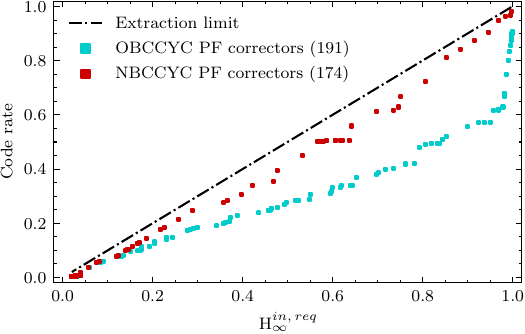}}
	\hfill
	\subfloat[\hspace*{-4em} \label{fig:b_CYC_Optimal_selection_fig_0.999}]{\includegraphics{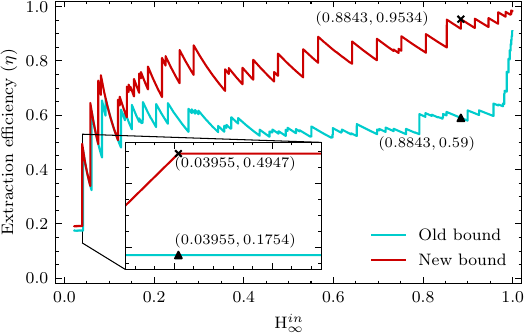}}
	\caption{Performances of optimal linear correctors from OBCCYC and NBCCYC for $\mathrm{H}^{out,\, 1}_\infty \geq 0.999$ and extraction efficiency according to the old and the new bound.} 
	\label{fig:CYC_Optimal_selection_fig_0.999}
\end{figure*}

\subsection{Practical Corrector Selection and Efficiency Comparisons}
\label{subsec:results}

In this work, we use $\mathrm{H}^{out,\, 1}_{\infty} = 0.999$, as it is the maximum between the requirement of the latest version of AIS-31 \cite{AIS31New} (0.98) and  NIST SP 800-90B \cite{Turan2018} upper bound for the min-entropy rate after non-cryptographic post-processing (0.999).
With this setting, we identified $24,221$ appropriate correctors from the OBC set, $15,873$ from the NBC set, $522$ from the OBCCYC set and $435$ from the NBCCYC set.

We evaluated the improvement in lowering $\mathrm{H}^{in,\, req}_\infty$ offered by the new bound by calculating the difference between the $\mathrm{H}^{in,\, req}_\infty$ values for $\mathrm{H}^{out,\, 1}_{\infty} = 0.999$ according to the new and the old bound for $9,908$ appropriate correctors common to both the OBC and NBC sets.
Our analysis revealed that the new bound yields a considerable relative improvement in $\mathrm{H}^{in,\, req}_\infty$ surpassing $15\, \%$ for most constructions. 
We found that the greatest absolute improvement is achieved for the Reed-Muller $\left[256, 93, 32\right]$ code-based corrector, for which the new bound lowers $\mathrm{H}^{in,\, req}_\infty$ from 0.854296 to 0.407964, while the largest relative improvement of $61.62 \,\%$ is obtained for the Reed-Muller $\left[512, 130, 64\right]$ corrector, as indicated in \figurename~\ref{fig:Max_entropy_diff}.
It is worthwhile to note that the old bound fails to guarantee that every output bit will have at least some entropy for $\mathrm{H}^{in}_\infty = 0.274447$ in the case of $\left[512, 130, 64\right]$ corrector and $\mathrm{H}^{in}_\infty = 0.407964$ in the case of $\left[256, 93, 32\right]$ corrector, by taking a conservative approach to calculating the output min-entropy rate $\mathrm{H}^{out,\, 1}_\infty = \mathrm{max} \left( \mathrm{H}^{out,\, tot}_\infty - k + 1,\, 0 \right)$. 
Even for correctors based on codes with very large minimum distances, such as the $\left[512, 10, 256\right]$ corrector, our bound still offers a discernible improvement of $0.01 \,\%$. 
This indicates that the state-of-the-art min-entropy bound for these correctors is already quite close to the new bound, underscoring that further improvements for the same $\mathrm{H}^{out,\, 1}_{\infty}$ are not feasible.

Appropriate correctors from OBC and NBC sets in a \textit{code rate - required input min-entropy} plane are shown in 
\figurename~\ref{fig:a_Optimal_selection_fig_0.999} -- \ref{fig:c_Optimal_selection_fig_0.999}. The dash-dotted lines show the theoretical extraction limit for $\mathrm{H}^{out,\, 1}_{\infty} = 0.999$, i.e., the highest possible code rate of $\nicefrac{\mathrm{H}^{in}_\infty}{\mathrm{H}^{out,\, 1}_{\infty}}$ for $\mathrm{H}^{in}_\infty < \mathrm{H}^{out,\, 1}_{\infty}$.
\figurename~\ref{fig:c_Optimal_selection_fig_0.999} displays optimal extracting (PF) correctors from both sets to examine the benefits of the new bound.
Although the NBC set of appropriate correctors is much smaller than its OBC counterpart, the optimal extracting solutions obtained by our bound always dominate over the solutions with the old bound.
Further, the new bound provides more optimal extracting correctors than the old one, though the correctors from the OBC set are more evenly spread.
Our analysis also revealed that the OBC set's optimal extracting correctors have the smallest $\mathrm{H}^{in,\, req}_{\infty}$ value of $0.022275$, whereas the NBC set's optimal extracting correctors have the smallest $\mathrm{H}^{in,\, req}_{\infty}$ value of $0.020351$. 
These results suggest that the new bound permits a marginally broader range of admissible raw bit min-entropies.

\figurename~\ref{fig:d_Optimal_selection_fig_0.999} shows the extraction efficiency for targeted $\mathrm{H}^{in}_\infty$ in the common range for both bounds -- $\left(0.022275, 0.999\right)$, by using the optimal extracting correctors selected according to the state-of-the-art and the new bound from the OBC and NBC sets, respectively.
The extraction efficiency is calculated using (\ref{eq:extraction_efficiency}).
For the old bound, we obtained $\mathrm{H}^{out,\, tot}_{\infty}$ as described in (\ref{eq:Lacharme_ineq_mod}), while for the new bound, we utilized (\ref{eq:Hmin_pp}). 
As indicated by peaks in the graph, extraction efficiency reaches local maxima for targeted min-entropies that coincide with $\mathrm{H}^{in,\, req}_\infty$ of the optimal extracting correctors.
Here, we observe that the extraction efficiencies for both bounds are consistently greater than $0.5$ starting from $\mathrm{H}^{in}_\infty = 0.08374$ and that the new bound extraction efficiency outperforms the old bound one for the entire input min-entropy range.
The largest absolute efficiency difference of $0.39668$ is reached for $\mathrm{H}^{in}_\infty = 0.76697$, while the highest relative efficiency increase of $130.56\, \%$ is achieved for $\mathrm{H}^{in}_\infty = 0.03947049$.
Additionally, we computed the average relative efficiency increase resulting from the new bound to be $41.2\, \%$, while starting from $\mathrm{H}^{in}_\infty = 0.1777221$ this relative increase consistently exceeds $20\, \%$.
The performances of optimal correctors from both sets for nine targeted input min-entropies are summarized in Table~\ref{tab:PF_codes_0.999}, together with constructions of corresponding correctors.

\begin{table}[t!]
	\caption{Optimal Linear Correctors Based on Cyclic Codes and Performances for $\mathrm{H}^{out,\, 1}_{\infty} \geq 0.999$}
	\label{tab:PF_codes_0.999_cyclic}
	\setlength\tabcolsep{2pt}
	\setlength\extrarowheight{2pt} 
	\begin{threeparttable}
		\begin{tabularx}{\linewidth}{@{\extracolsep{4pt}} *{5}{C} @{}}
			\toprule
			\multirow{3}{*}{Target $\mathrm{H}^{in}_{\infty}$} 
			& \multicolumn{2}{c}{Corrector construction} & \multicolumn{2}{c}{Extraction efficiency ($\eta$)} \\ 
			\cline{2-3} \cline{4-5} & OBCCYC & NBCCYC & Old bound & New bound\\
			\midrule
			\multirow{2}{*}{$0.1$}   & \multirow{2}{*}{$\left[511, 31, 219\right]$\tnote{a}}    & \multirow{2}{*}{$\left[511, 31, 219\right]$\tnote{a}}  & \multirow{2}{*}{$0.60665234$}    & $0.60665360$ $\left(+0.0002\% \right)$ \\ 
            \multirow{2}{*}{$0.2$}   & \multirow{2}{*}{$\left[255, 29, 95\right]$\tnote{a,b}}    & \multirow{2}{*}{$\left[255, 37, 91\right]$\tnote{a,b}}  & \multirow{2}{*}{$0.56862$}    & $0.72549$ $\left(+27.59\% \right)$ \\ 
            \multirow{2}{*}{$0.3$}   & \multirow{2}{*}{$\left[255, 47, 85\right]$\tnote{a,b}}    & \multirow{2}{*}{$\left[255, 63, 63\right]$\tnote{a,b}}  & \multirow{2}{*}{$0.61437$}    & $0.82353$ $\left(+34.04\% \right)$ \\ 
            \multirow{2}{*}{$0.4$}   & \multirow{2}{*}{$\left[127, 29, 43\right]$\tnote{*}}    & \multirow{2}{*}{$\left[117, 36, 32\right]$\tnote{*, c}}  & \multirow{2}{*}{$0.57086$}    & $0.76921$ $\left(+34.75\% \right)$ \\ 
            \multirow{2}{*}{$0.5$}   & \multirow{2}{*}{$\left[127, 35, 36\right]$\tnote{*}}    & \multirow{2}{*}{$\left[127, 50, 27\right]$\tnote{a,b}}  & \multirow{2}{*}{$0.55117$}    & $0.78740$ $\left(+42.86\% \right)$ \\ 
            \multirow{2}{*}{$0.6$}   & \multirow{2}{*}{$\left[127, 42, 32\right]$\tnote{*}}    & \multirow{2}{*}{$\left[127, 64, 19\right]$\tnote{d}}  & \multirow{2}{*}{$0.55117$}    & $0.83989$ $\left(+52.38\% \right)$ \\ 
            \multirow{2}{*}{$0.7$}   & \multirow{2}{*}{$\left[55, 21, 15\right]$\tnote{*}}    & \multirow{2}{*}{$\left[127, 78, 15\right]$\tnote{a,b}}  & \multirow{2}{*}{$0.54543$}    & $0.87738$ $\left(+60.86\% \right)$ \\ 
            \multirow{2}{*}{$0.8$}   & \multirow{2}{*}{$\left[23, 11, 8\right]$\tnote{*}}    & \multirow{2}{*}{$\left[127, 85, 13\right]$\tnote{a,b}}  & \multirow{2}{*}{$0.59779$}    & $0.83661$ $\left(+39.95\% \right)$ \\ 
            \multirow{2}{*}{$0.9$}   & \multirow{2}{*}{$\left[63, 35, 12\right]$\tnote{*}}    & \multirow{2}{*}{$\left[255, 215, 11\right]$\tnote{a,b}}  & \multirow{2}{*}{$0.61727$}    & $0.93682$ $\left(+51.77\% \right)$ \\ 
			\bottomrule
		\end{tabularx}
		\begin{tablenotes}[para, flushleft]\footnotesize
\item[*] BKLC code from \cite{Grassl_codetables}
            \item[a] BCH code from \cite{Lin2004}
\item[b] BCH code from \cite{Teradaweb}
            \item[c] Code from \cite{Schomaker92}
            \item[d] Quadratic residue code from \cite{oeis}
		\end{tablenotes}
	\end{threeparttable}
\end{table}

The code rates of the optimal extracting correctors based only on cyclic codes from OBCCYC and NBCCYC sets versus their $\mathrm{H}^{in,\, req}_{\infty}$ is depicted in \figurename~\ref{fig:a_CYC_Optimal_selection_fig_0.999}.
In this case, there are fewer optimal correctors from the NBCCYC set, but we found that the new bound still provides a narrowly larger range of admissible input min-entropies, as the values of the smallest $\mathrm{H}^{in,\, req}_{\infty}$ for correctors in OBCCYC and NBCCYC sets are identical to the ones in OBC and NBC sets, respectively.
Based on the results shown in the plot of \figurename~\ref{fig:b_CYC_Optimal_selection_fig_0.999}, which displays the relationship between the extraction efficiency and the targeted $\mathrm{H}^{in}_{\infty}$, it is evident that the extraction efficiency achieved with the new bound-selected cyclic correctors always surpasses that of the old bound-selected cyclic correctors for all targeted $\mathrm{H}^{in}_{\infty}$.
Notably, the maximum relative efficiency increase of $182.04\, \%$ achieved for $\mathrm{H}^{in}_{\infty} = 0.03955041$ in this case is higher than the increase observed without imposing the cyclicity restriction.
Table~\ref{tab:PF_codes_0.999_cyclic} summarizes the performances and constructions of optimal correctors based on cyclic codes for nine targeted input min-entropies.
Optimal extracting corrector constructions from all sets and their $\mathrm{H}^{in,\, req}_{\infty}$ are available in our online repository \cite{OurGithub}.

\subsection{Implementation Cost Criterion}

\begin{figure}[!t]
	\centering
\includegraphics[scale=1]{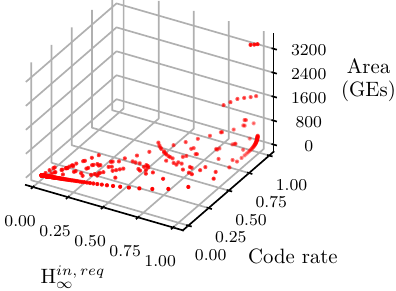}
	\caption{Optimal area-efficient cyclic code-based correctors.}
	\label{fig:Implementation_cost}
\end{figure}

As a final selection criterion, we take an estimation of the implementation cost (chip area) of the correctors based on cyclic codes.
Cyclic codes possess a distinct structure that results in a simplified implementation compared to general codes. 
Our objective is to find a balance between the code rate, required input min-entropy, and the area the corrector based on cyclic code would occupy. 
In doing so, we ensure that the chosen correctors not only provide a small reduction in throughput and high extraction efficiency but are also practical for real-world applications.

To evaluate the implementation cost of each corrector in the NBCCYC set, without  including a controller counter, we estimate the number of gate equivalents (GEs).
The area of each corrector is assessed based on two distinct implementation methods, utilizing the generator and parity-check polynomials of the corresponding code, as delineated in \cite{Kwok2011}. 
We employ XOR2\_1 and DFFR\_X1 gates from the NanGate 45 nm open standard-cell library \cite{NanGateRef}.
Each XOR2\_1 gate consumes 2 GEs, while the DFFR\_X1 gate utilizes 6.67 GEs. 
Here, one GE corresponds to the size of a NAND2\_X1 gate.
We first calculate the area of each corrector using both implementation flavors. Subsequently, for each individual corrector, we select the implementation yielding the smaller area. 
We then conduct a three-dimensional optimization to derive a set of optimal area efficiency correctors.
A corrector based on cyclic code is \textit{optimal area-efficient} if there are no other codes in NBCCYC that concurrently exhibit a higher code rate, equal or lower $\mathrm{H}^{in,\, req}_\infty$, and a smaller area.

The 434 optimal area-efficient correctors that we found are displayed in \figurename~\ref{fig:Implementation_cost}. 
Table~\ref{tab:PF_3D_codes_0.999_cyclic} provides an overview of the constructions and performances of these correctors for nine targeted input min-entropies. 
Comparing these correctors to the correctors found with the new bound listed in Table~\ref{tab:PF_codes_0.999_cyclic}, it is immediately evident that the correctors in Table~\ref{tab:PF_3D_codes_0.999_cyclic} exhibit significantly lower extraction efficiency, particularly for $\mathrm{H}^{in,\, req}_\infty = 0.1$ and $\mathrm{H}^{in,\, req}_\infty = 0.4$. 
However, these constructions require only 8.67 GEs, whereas correctors based on $\left[511, 31, 219\right]$ and $\left[117, 36, 32\right]$ codes require 224.77 GEs and 272.12 GEs, respectively.
On the other hand, for $\mathrm{H}^{in,\, req}_\infty = 0.5$, the efficiency of the $\left[87, 31, 22\right]$ corrector differs from that of the $\left[127, 50, 27\right]$ corrector by only 0.0748, while consuming much less area: 375.50 GEs vs 244.77 GEs. 
The estimated implementation costs for all correctors from the NBCCYC set are also available in \cite{OurGithub}.

\begin{table}[t!]
	\caption{Optimal Area-Efficient Linear Correctors Based on Cyclic Codes and Performances for $\mathrm{H}^{out,\, 1}_{\infty} \geq 0.999$ (New Bound)}
	\label{tab:PF_3D_codes_0.999_cyclic}
	\setlength\tabcolsep{2pt}
	\setlength\extrarowheight{2pt} 
	\begin{threeparttable}
		\begin{tabularx}{\linewidth}{@{\extracolsep{2pt}} *{4}{C} @{}}
			\toprule
\multirow{2}{*}{Target $\mathrm{H}^{in}_{\infty}$} & Corrector construction & Extraction efficiency ($\eta$) & {Area \hspace{10pt} (NanGate 45 nm)} \\
			\midrule
			$0.1$   & $\left[51, 1, 51\right]$\tnote{}    & $0.1959$ & $8.67$ GEs \\ 
			$0.2$   & $\left[127, 15, 55\right]$\tnote{}    & $0.5905$ & $122.05$ GEs \\ 
			$0.3$   & $\left[63, 9, 28\right]$\tnote{}    & $0.4762$ & $68.03$ GEs \\ 
			$0.4$   & $\left[11, 1, 11\right]$\tnote{}    & $0.2270$ & $8.67$ GEs \\ 
			$0.5$   & $\left[87, 31, 22\right]$\tnote{}    & $0.7126$ & $244.77$ GEs \\ 
			$0.6$   & $\left[127, 64, 21\right]$\tnote{}    & $0.8399$ & $484.88$ GEs \\ 
			$0.7$   & $\left[15, 5, 7\right]$\tnote{}    & $0.4761$ & $39.35$ GEs \\ 
			$0.8$   & $\left[23, 12, 7\right]$\tnote{}    & $0.6521$ & $94.04$ GEs \\ 
			$0.9$   & $\left[31, 21, 5\right]$\tnote{}    & $0.7527$ & $162.07$ GEs \\ 
			\bottomrule
		\end{tabularx}
\end{threeparttable}
\end{table}

\begin{table}[t!]
	\caption{Implementation Cost Comparisons for Different Post-processing Algorithms}
	\label{tab:Implementation_comparisons}
	\setlength\tabcolsep{2pt}
	\setlength\extrarowheight{2pt} 
	\begin{threeparttable}
		\begin{tabularx}{\linewidth}{@{\extracolsep{2pt}} *{4}{C} @{}}
			\toprule
			Post-processing & Reference & Technology & Area \\
			\midrule
			Keccak-$f$ [1600]   & \cite{Knichel2022}\tnote{a}    & NanGate 45 nm & $31,361$ GEs \\ 
			Keccak-$f$ [1600]   & \cite{Bilgin2014}\tnote{a}    & NanGate 45 nm & $28,100$ GEs \\ 
			SHA-256   & \cite{Baldanzi2020}\tnote{\hspace{3pt}}    & NanGate 45 nm & $15,000$ GEs \\ 
			Keccak-$f$ [1600]   & \cite{Bilgin2014}\tnote{b}    & NanGate 45 nm & $12,800$ GEs \\ 
			SHA-256   & \cite{Kim2009}\tnote{\hspace{3pt}}  &STD110  0.25$\,\mu m$& $8,588$ GEs \\ 
			Keccak-$f$ [1600]   & \cite{Pessl2013}\tnote{b}    & UMC 0.13 $\mu m$ & $5,522$ GEs \\ 
			Linear corrector $\left[511, 484, 7\right]$  & \multirow{2}{*}{This work\tnote{c}}    & \multirow{2}{*}{NanGate 45 nm} & \multirow{2}{*}{$3443.04$ GEs} \\  
			\bottomrule
		\end{tabularx}
		\begin{tablenotes}[para, flushleft]\footnotesize
			\item[a] round-based
			\item[b] serial (slice-based)
			\item[c] largest optimal area-efficient NBCCYC corrector
		\end{tablenotes}
	\end{threeparttable}
\end{table}

Table~\ref{tab:Implementation_comparisons} shows the area usage (in GEs) for the largest optimal area-efficient linear corrector $\left[511, 484, 7\right]$ and several implementations of two NIST-approved cryptographic hash functions that can be used for post-processing (conditioning) \cite{Turan2018} -- SHA-3 (based on Keccak-$f$ [1600]) and SHA-256.
It can be observed that the areas of various implementations of Keccak-$f$ [1600] and SHA-256 vary significantly due to the technology and architectural choices.
However, even the implementation of the largest linear corrector $\left[511, 484, 7\right]$ from our work demonstrates a remarkable reduction in the area footprint, consuming only $3443.04$ GEs. 
This represents a considerable saving in comparison to the cryptographic post-processing algorithms.
Further, when considering only implementations using identical technology -- Nangate 45 nm, the $\left[511, 484, 7\right]$ corrector is more than three times smaller than the most area-efficient implementation of Keccak-$f$ [1600].

 \section{Conclusion}
\label{Sec:Conclusions}

In this paper, we have presented a novel tight bound on the output min-entropy of linear correctors based on the weight distribution of the corresponding binary linear code. 
Our proposed bound, which relies on the code's weight distribution, enables more efficient use of linear correctors than the old bound, which only requires knowledge of the code's minimum distance. 
We have demonstrated how the new bound can be used to select an optimal extracting corrector that meets output min-entropy rate requirements and maximizes throughput.
Moreover, we have made publicly available optimal constructions for general correctors and correctors based on cyclic codes for $\mathrm{H}^{out,\, 1}_\infty = 0.999$, allowing for easy implementation and integration into existing TRNG designs.
Our findings indicate a potential for advancements in optimal extracting solutions through further research in characterizing binary linear codes' weight distributions. Future work will concentrate on constructing tight output min-entropy bounds for a wider spectrum of non-IID noise sources and, potentially, non-linear correctors.

\section*{Acknowledgment}

The authors would like to thank the anonymous reviewers for their useful feedback and highlighting the connection between our findings and those presented in the work by Redinbo \cite{Redinbo1973} using Fourier methods.

\bibliographystyle{IEEEtran}
\bibliography{IEEEabrv, correctors_paper_bibl}

\end{document}